\documentclass[journal,12pt,onecolumn,draftclsnofoot,]{IEEEtran}

\usepackage[pdftex]{graphicx}

\usepackage[T1]{fontenc}
\usepackage{color}
\usepackage{varioref}
\usepackage{textcomp}
\usepackage{amsthm}
\usepackage{amsmath}
\usepackage{amssymb}
\usepackage{graphicx}
\usepackage{setspace}
\usepackage{esint}
\usepackage{algorithm}
\usepackage{algpseudocode}
\usepackage{pifont}
\usepackage{amsmath}

\usepackage{enumitem}
\makeatletter

\newcommand{\Rmnum}[1]{\expandafter\@slowromancap\romannumeral #1@}
\makeatother

\newtheorem{theorem}{Theorem}

\newtheorem{remark}{Remark}
\newtheorem{corollary}{Corollary}

\providecommand{\propositionname}{Proposition}

\ifCLASSINFOpdf

\else

\fi

\usepackage[T1]{fontenc}
\usepackage{etoolbox}

\makeatletter
\patchcmd{\maketitle}{\@fnsymbol}{\@alph}{}{}  
\makeatother

\title{Coded Caching for a Large Number Of Users}
\author{
  Mohammad Mohammadi Amiri, Qianqian Yang, and\thanks{The authors are with Imperial College London, London SW7 2AZ, U.K. (e-mail: m.mohammadi-amiri15@imperial.ac.uk; q.yang14@imperial.ac.uk; d.gunduz@imperial.ac.uk).}
  \and
  Deniz G\"und\"uz
}
\date{}

\begin{document}

\maketitle

\begin{abstract}
Information theoretic analysis of a coded caching system is considered, in which a server with a database of $N$ equal-size files, each $F$ bits long, serves $K$ users. Each user is assumed to have a local cache that can store $M$ files, i.e., capacity of $MF$ bits. \textit{Proactive caching} to user terminals is considered, in which the caches are filled by the server in advance during the \textit{placement phase}, without knowing the user requests. Each user requests a single file, and all the requests are satisfied simultaneously through a shared error-free link during the \textit{delivery phase}. 


First, \textit{centralized coded caching} is studied assuming both the number and the identity of the active users in the delivery phase are known by the server during the placement phase. A novel group-based centralized coded caching (GBC) scheme is proposed for a cache capacity of $M = N/K$. It is shown that this scheme achieves a smaller delivery rate than all the known schemes in the literature. The improvement is then extended to a wider range of cache capacities through memory-sharing between the proposed scheme and other known schemes in the literature. Next, the proposed centralized coded caching idea is exploited in the \textit{decentralized} setting, in which the identities of the users that participate in the delivery phase are assumed to be unknown during the placement phase. It is shown that the proposed decentralized caching scheme also achieves a delivery rate smaller than the state-of-the-art. Numerical simulations are also presented to corroborate our theoretical results. 
\\ 
\end{abstract}


\section{Introduction}\label{Intro}
There has been a recent revival of interest in \textit{content caching}, particularly focusing on wireless networks. This interest stems from a very practical problem: exponential growth in mobile traffic cannot be matched by the increase in the spectral efficiency of wireless networks. This, in turn, leads to congestion in the radio access as well as the backhaul links, and increased delay and outages for users. \textit{Proactively caching} popular contents at the network edge during off-peak hours has been recently proposed as a potential remedy for this problem (see \cite{AlmerothCacing,GolrezaeiFemtocaching,MaddahAliCentralized,GregoryDtoD,JiArXivNonuniform}, and references therein). Proactive caching shifts traffic from peak to off-peak hours, reduces latency for users, and potentially provides energy savings. In this paper, we focus on the coded caching model proposed in \cite{MaddahAliCentralized}, which considers a single server holding a database of $N$ popular contents of equal size ($F$ bits), serving a set of $K$ users, which have local storage space, sufficient to hold $M$ files, that can be used to proactively cache content during off-peak hours.

Caching in this model consists of two distinct phases: In the first phase, which takes place during off-peak periods, i.e., when the network is not congested, caches at user terminals are filled by the server. This first phase is called the \emph{placement phase}. The only constraint on the data transmitted and stored in a user cache in this phase is the cache capacity. However, due to the ``proactive'' nature of the cache placement, it is carried out without the knowledge of the particular user requests. A shared communication channel is considered to be available from the server to all the users during the peak traffic period. Once the user demands are revealed in this period, the server broadcasts additional information over the common error-free channel in order to satisfy all the user requests simultaneously. This constitutes the \emph{delivery phase}. Since the delivery phase takes place during peak traffic period, the goal is to minimize the rate of transmission over the shared link, called the \textit{delivery rate}, by exploiting the contents that are available the caches.   

Over the past decade, research on caching has mainly focused on the placement phase; the goal has been to decide which contents to cache, typically at a server that serves many users, by anticipating future demands based on the history (see \cite{baev2008approximation, BorstCaching, Blasco:ISIT:14}, and references therein). In our model, due to the uniform popularity of the files in the database, this \textit{uncoded} caching approach simply stores an equal portion of each file (i.e., fraction of $M/N$), in each cache. The delivery rate of the uncoded caching scheme, when the user requests are as distinct as possible, i.e., the worst-case rate, is given by 
\begin{align}\label{rate_uncoded}
    R_{\rm{U}}(M) &= K \cdot \left(1 - \frac{M}{N} \right) \cdot \min\left\{1, \frac{N}{K} \right\}.
\end{align}
The gain from this conventional uncoded caching approach, compared to not having any caches at the users, i.e., the factor $(1-M/N)$ in (\ref{rate_uncoded}), derives mainly from the availability of popular contents locally.

On the other hand, the \emph{coded caching} approach proposed in \cite{MaddahAliCentralized} jointly designs the placement and delivery phases in a \textit{centralized} manner; that is, a central server, which knows the number and identity of the users in the system, can arrange the placement and delivery phases in coordination across the users. This coded caching approach, thanks to the presence of a common broadcast channel to all the users, can significantly reduce the backhaul traffic over conventional uncoded caching by creating multicasting opportunities, even when users request different files. The following delivery rate is achieved by the centralized coded caching scheme proposed in~\cite{MaddahAliCentralized}, referred to as the MAN scheme in the rest of the paper:
\begin{align}\label{rate_coded_MAN}
    R^{\rm{C}}_{\rm{MAN}}(M) &= K \cdot (1 - M/N) \cdot \min\left\{\frac{1}{1+KM/N}, \frac{N}{K}\right\}.
\end{align}

Following the seminal work of \cite{MaddahAliCentralized}, further research on centralized coded caching have been carried out recently to reduce the required delivery rate. Authors in \cite{ZhiChenXOR} have proposed a coded placement phase, referred to as the CFL scheme in this paper. When the number of users, $K$, is at least as large as the number of files in the system, i.e., $K \ge N$, the following delivery rate can be achieved for a cache capacity of $M=1/K$ by the CFL scheme:
\begin{align}\label{rate_coded_CFL}
    R^{\rm{C}}_{\rm{CFL}} \left( \frac{1}{K} \right) &= N \left( 1 - \frac{1}{K} \right).
\end{align}
It is also shown in \cite{ZhiChenXOR} that the proposed scheme is indeed optimal for small cache capacities; that is, when $M \leq 1/K$.

By employing coded multicasting opportunities across users with the same demand, an improved centralized caching scheme for any cache capacity $0 \le M \le N$ is presented in \cite{KaiWanUncodedCaching} for the special case $K > N=2$. Again, considering more users in the system than the files in the database, i.e., $K> N$, the authors in \cite{MohammadDenizTCom} have further reduced the delivery rate by exploiting a novel coded placement scheme, when $4 \le N < K \le 3N/2$, and $N$ and $K$ are not relatively prime. 

A theoretical lower bound on the delivery rate allows us to quantify how far the proposed caching schemes perform compared to the optimal performance. In addition to the cut-set bound studied in \cite{MaddahAliCentralized}, a tighter lower bound is derived in \cite{SenguptaCaching}. For the special case of $N=K=3$, a lower bound is derived through a novel computational approach in \cite{TianCaching}. A labeling problem in a directed tree, the complexity of which grows exponentially with the number of users, is considered in \cite{GhasemiCachingLowerBound}, to characterize another lower bound. However, none of these bounds are tight in general, and the optimal delivery rate-cache capacity trade-off for a caching system remains an open problem. 

In contrast to the above centralized setting, in many practical scenarios, the identity and number of active users in the delivery phase are not known in advance, and there may not be a central authority to coordinate the cache contents across users. Maddah-Ali and Niesen considered the corresponding \textit{decentralized} coded caching problem, and showed that coded caching can still help reduce the delivery rate \cite{MaddahAliDecentralized}. In the decentralized scenario, while the multicasting opportunities cannot be designed and maximized in advance, they will still appear even if the bits of files are randomly cached in the placement phase. The following delivery rate is achievable by the proposed decentralized caching scheme in \cite{MaddahAliDecentralized}:
\begin{align}\label{rate_dec_MAN}
    R^{\rm{D}}_{\rm{MAN}}(M) &= K \cdot \left(1 - \frac{M}{N} \right) \cdot \min\left\{\frac{N}{KM}\big(1-(1-M/N)^K\big), \frac{N}{K}\right\}.
\end{align}
We note that the delivery rate in (\ref{rate_dec_MAN}) provides the same local caching gain as in the centralized setting, i.e., the factor $(1-M/N)$, as well as an additional global gain thanks to coding. The centralized caching schemes proposed in \cite{KaiWanUncodedCaching} is also extended to the decentralized setting, and it is shown to reduce the required delivery rate. 

Coded caching have been studied under many other network settings in the recent years, such as online coded caching \cite{PedarsaniOnlineCaching}, multi-layer coded caching \cite{KaramchandaniHierarchical}, caching files with nonuniform distributions \cite{NiesenNonuniform, JiArXivNonuniform} and distinct sizes \cite{ZhangDistinctFileSizes}, user with distinct cache capacities \cite{WangHeterogenous}. It has also been extended to lossy reconstructions of cached files in \cite{QianDenizLossy, ElzaDistortionMemoryTradeoff, TimoDistortionCaching}. Several works have considered noisy channels rather than an error-free shared link in the delivery phase \cite{MaddahAliInterferenceJournal,NaderializadehMaddahAliInterference, TimoErasureChannel}, as well as delivery over a fading channel \cite{HuangFadingChannelcodedcaching}.  

In this paper, we study coded caching in both the centralized and decentralized settings. We propose a novel \textit{group-based caching scheme}, and show that it improves the caching gain in both scenarios, particularly when the number of users in the system is larger than the number of files, i.e., when $K > N$. Note that, this setting may be applicable for the distribution of extremely popular files that become viral over the Internet, and are requested by a large number of users over a relatively short period of time, or for the distribution of various software updates to users. 

Our main contributions can be summarized as below:

\begin{enumerate}[label=\arabic*)]
\item  In the centralized setting, we propose a novel \textit{group-based centralized} (GBC) coded caching scheme for a cache size of $M = N/K$ at the users. It is shown that the GBC scheme achieves a lower delivery rate compared to the state-of-the-art results when $K>N\geq3$. This improvement can be further extended to other cache sizes through memory-sharing with the schemes proposed in \cite{MaddahAliCentralized} and \cite{ZhiChenXOR}.

\item By employing the same group-based caching idea in the decentralized setting, we introduce the \textit{group-based decentralized caching} (GBD) scheme. The GBD scheme is shown to achieve a smaller delivery rate compared to the schemes presented in \cite{MaddahAliDecentralized} and \cite{KaiWanUncodedCaching} for $K>N$. 

\item We provide numerical results validating the superiority of the proposed group-based coded caching schemes compared to the state-of-the-art in both the centralized and decentralized settings. 
\end{enumerate}

The rest of the paper is organized as follows. We present
the system model and relevant previous results in Section \ref{SystemModel}. The group-based centralized coded caching scheme is introduced in Section \ref{CentralizedScheme} for the centralized setting. In Section \ref{DecentralizedScheme} it is extended to the decentralized scenario. In both Sections \ref{CentralizedScheme} and \ref{DecentralizedScheme}, the derivation of the corresponding delivery rates are complemented with the analytical and numerical comparison of the achieved delivery rates with the state-of-the-art results. The proofs of our main results can be found in Appendix. We conclude the paper in Section \ref{Conclusion}.

\textit{Notations:} The sets of integers, real numbers, and positive real numbers are denoted by $\mathcal Z$, $\mathcal R$, and $\mathcal R^+$, respectively. For two integers $i \le j$, $\left[ {i:j} \right]$ represents the set of integers $\left\{ {i, i + 1, \ldots , j} \right\}$. If $i>j$, then $\left[ {i:j} \right] \buildrel \Delta \over = \left\{ \emptyset  \right\}$. The binomial coefficient ``$n$ choose $k$'' is denoted by $\binom{n}{k}$. Notation $ \oplus $ refers to the bitwise XOR operation. We use $\left| \cdot \right|$ to indicate the length of a binary sequence, or the cardinality of a set. 

\section{System Model and Previous Results}\label{SystemModel}
Consider a server which has $N$ popular files, denoted by $W_1, W_2, ..., W_N$, each of length $F$ bits, in its database. The files are assumed to be independent of each other and distributed uniformly over the set $\left[ 1:2^F \right]$. There are $K$ users in the system, denoted by $U_1, U_2, ..., U_K$, and each user is equipped with a local cache of capacity $MF$ bits. 

The caching system operates in two distinct phases. In the initial \textit{placement phase}, cache of each user is filled by the server. The contents of the cache of user $U_k$ at the end of the placement phase is denoted by $Z_k$, for $k=1,...,K$. The user demands are revealed after the placement phase. Each user requests a single file from among the available files in the database, and the demand of user $U_k$ is denoted by $d_k$, where ${d_k} \in \left[ 1:N \right]$, $\forall k$. In the \textit{delivery phase} that follows, a common message $X$ is transmitted by the server to the users over an error-free shared link to satisfy all the user requests simultaneously. 

\textbf{Definition.} An $(M,R,F)$ \textit{caching and delivery code} consists of
\begin{enumerate}[label=\roman*.]
\item $K$ caching functions:
\begin{equation}\label{S1} {\phi _k}:{\left[ {1:{2^F}} \right]^N} \to \left[ {1:{2^{\left\lfloor {FM} \right\rfloor }}} \right], \quad k=1, ..., K,
\end{equation}
each of which maps the database $\left\{ W_1, W_2, ..., W_N \right\}$ to the cache content $Z_k$ of user $U_k$, i.e., $Z_k = \phi _k \left( W_1, W_2, ..., W_N \right)$;
\item one coded delivery function:
\begin{equation}\label{S2} \psi :{\left[ {1:{2^F}} \right]^N} \times {\left[ 1:N \right]^K} \to \left[ {1:{2^{\left\lfloor {FR} \right\rfloor }}} \right],
\end{equation}
which generates the common message $X$ to be delivered through the shared link as a function of the database $\left\{ W_1, W_2, ..., W_N \right\}$ and the user demands $\left\{ d_1, d_2, ..., d_K \right\}$, i.e., $X=\psi \left(W_1, W_2, ..., W_N, d_1, d_2, ..., d_K\right)$;
\item $K$ decoding functions:
\begin{equation}\label{S3} {\mu _k}:\left[ {1:{2^{\left\lfloor {FM} \right\rfloor }}} \right] \times \left[ {1:{2^{\left\lfloor {FR} \right\rfloor }}} \right] \times \left[ 1:N \right]^K \to \left[ {1:{2^F}} \right], \quad k=1, ..., K,
\end{equation}
each of which reconstructs file ${\hat W_{{d_k}}}$ at user $k$ as a function of the cache content $Z_k$, the common message $X$ delivered over the shared link, and the user demands $\left\{ d_1, d_2, ..., d_K \right\}$, i.e., ${\hat W_{{d_k}}} = {\mu _k}\left(Z_k, X, d_1, d_2, ..., d_K\right)$.
\end{enumerate} 

In the above definition, $M$ represents the normalized cache capacity of each user (normalized by $F$), and $R$ denotes the transmission rate during the delivery phase.

\textbf{Definition.} The probability of error of an $(M,R,F)$ caching and delivery code is defined as
\begin{equation}\label{S4} {P_e} \buildrel \Delta \over = \mathop {\max }\limits_{\left( {{d_1},{d_2},...,{d_K}} \right)} \Pr \left\{ {\mathop  \bigcup \limits_{k = 1}^K \left\{ {{{\hat W}_{d_k}} \ne {W_{{d_k}}}} \right\}} \right\}.
\end{equation}

\textbf{Definition.} The delivery rate-cache capacity pair $(M,R)$ is \textit{achievable} if for every $\varepsilon  > 0$ an $(M,R,F)$ caching and delivery code can be found with significantly large $F$, that has a probability of error less than $\varepsilon$, i.e., $P_e < \varepsilon$.

\begin{figure}[!t]
\centering
\includegraphics[scale=0.7]{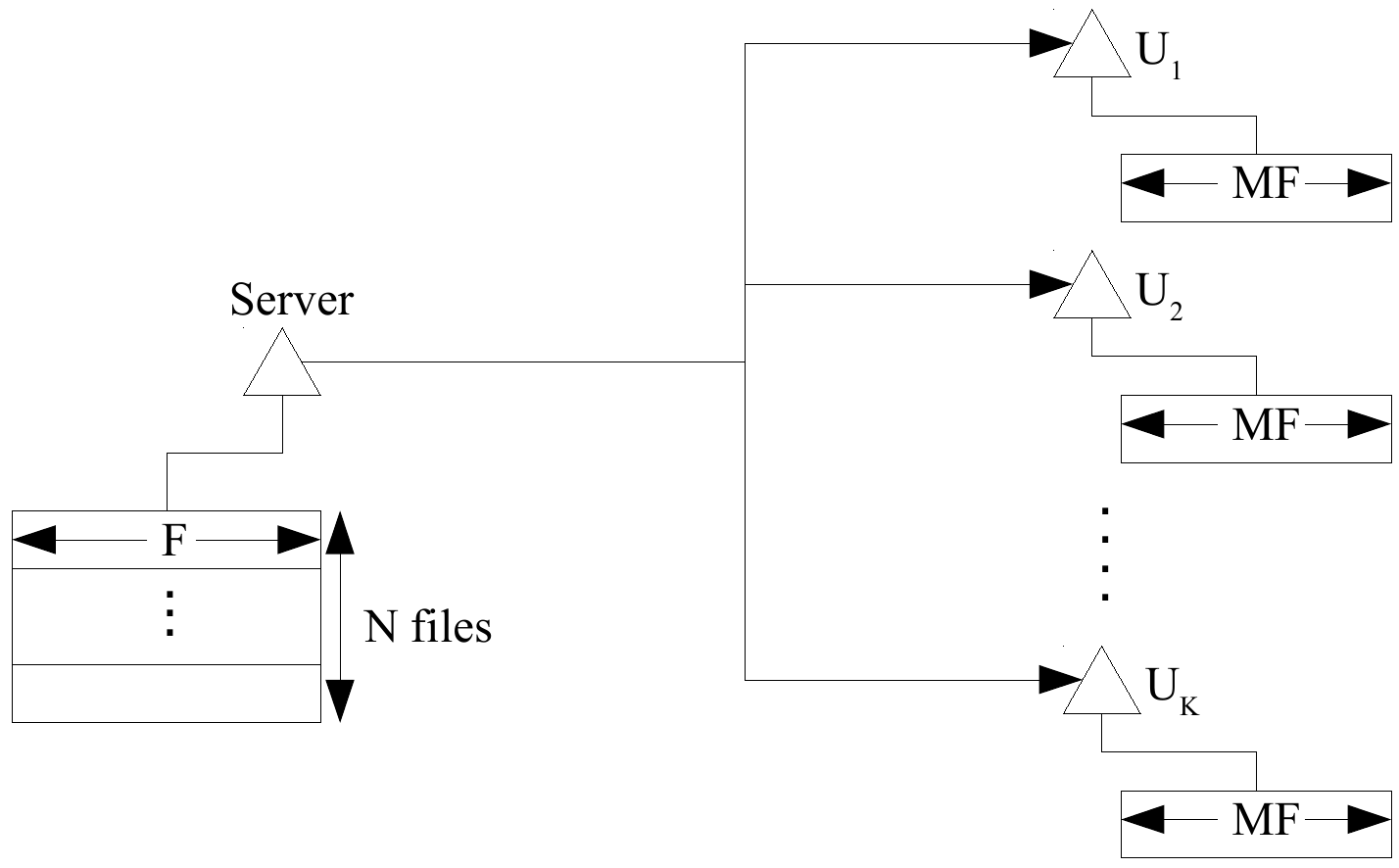}
\caption{Illustration of a caching system including a server holding $N$ popular files, each $F$ bits long, in its database, serving contents to $K$ users, each equipped with a local cache of capacity $MF$ bits, through an error-free shared link.} 
\label{System_Model}
\end{figure} 

Naturally, there is a trade-off between the cache capacity of the users, and the required minimum delivery rate. Our goal is to define and characterize this trade-off rigorously. The delivery rate-cache capacity trade-off for the caching network depicted in Fig. \ref{System_Model}, with the caches of equal capacity $M$, is denoted by $R^{\rm{C}}(M)$, and defined as
\begin{equation}\label{S5} {R^{\rm{C}}}\left( M \right) \buildrel \Delta \over = \inf \left\{ {R:\left( {M,R} \right) \mbox{ is achievable} } \right\}.
\end{equation}
For example, when there is no cache available at the users, i.e., $M=0$, the worst case of user demands corresponds to users requesting as distinct files as possible; and all the requested files should be delivered by the server over the shared link; we have $R^{\rm{C}}(0)=\min \left\{ {N,K} \right\}$. On the other hand, when the cache capacity is large enough to store all the files, i.e., when $M=N$, all the files can be made available locally to all the users, and no content needs to be sent during the delivery phase, and $R^{\rm{C}}(N)=0$. Our goal is to characterize the delivery rate-cache capacity trade-off for all cache capacities $0 < M < N$ for any given number of users $K$.   

For a centralized caching system, we denote the best achievable delivery rate-cache capacity trade-off in the literature by $R^{\rm{C}}_b(M)$. For $N > K$, the best known delivery rate is achieved by the scheme proposed in \cite{MaddahAliCentralized}, for any cache capacity satisfying $0 \le M \le N$, i.e., we have $R^{\rm{C}}_b(M) = R^{\rm{C}}_{\rm{MAN}}(M)$. On the other hand, for $N \le K$, we define the following set of $\left( N,K \right)$ corresponding to the scenario presented in \cite{MohammadDenizTCom} to find the best achievable scheme in the literature:
\begin{equation}\label{PaticularSetOurs} \zeta \buildrel \Delta \over = \left\{ {\left( {N,K} \right):4 \le N < K \le \frac{{3N}}{2},\mbox{$N$ and $K$ have a common divisor $c>1$}} \right\}.
\end{equation}
When $N \le K$, two following cases are considered to characterize $R_b^{\rm{C}}(M)$.  

\begin{enumerate}[label=\bfseries Case \arabic*:,align=left]
\item $\left( {N,K} \right) \in \zeta$. The best delivery rate in the literature for Case 1 can be characterized as follows. For cache capacity $M = 1/K$, the CFL scheme achieves the optimal delivery rate \cite{ZhiChenXOR}, while for $M = (N-1)/K$, the coded caching scheme presented in \cite{MohammadDenizTCom}, which will be referred to as AG, achieves the best known delivery rate. For cache capacities $M=tN/K$, the MAN scheme should be utilized, where if $K < 3N/2$, then $t \in \left[ 1:K \right]$, and if $K = 3N/2$, then $t \in \left[ 2:K \right]$. For any other cache capacities $0 \le M \le N$, the lower convex envelope of the mentioned points can be achieved through memory-sharing. 
\item $N \le K$ and $\left( {N,K} \right) \notin \zeta$. The best delivery rate in the literature is achieved by memory-sharing between the MAN and CFL schemes. In this case, the CFL scheme again achieves the optimal delivery rate for a cache capacity of $M = 1/K$. To find the best achievable scheme for $M = N/K$, we define, for $t \in \left[ 1:K \right]$,
\begin{align}\label{fNKt} f(N,K,t) \triangleq \frac{{\left( {N - 1} \right)\left( {K - t} \right)}}{{\left( {t + 1} \right)\left( {tN - 1} \right)}} + {N^2}\left( {1 - \frac{1}{K}} \right)\left( {\frac{{t - 1}}{{tN - 1}}} \right),
\end{align}
and
\begin{equation}\label{optimumt} {t^*} \buildrel \Delta \over = \mathop {\arg \min }\limits_{t \in \left[ {1:K} \right]} f(N,K,t). 
\end{equation}
Then, we have 
\begin{equation}\label{BestDeliveryRatePoint} R_b^{\rm{C}}\left(N/K \right) = f(N,K,t^*).
\end{equation} 
For cache capacities $M = lN/K$, where $l \in \left[ t^*:K \right]$, the MAN scheme should be employed. The best delivery rate for other cache capacities is the lowest convex envelope of the points $M = l/K$ and $M = lN/K$, for $l \in \left[ t^*:K \right]$, which can be achieved by memory-sharing.
\end{enumerate}

\section{Centralized Coded Caching}\label{CentralizedScheme}
In this section, the placement and delivery phases of the proposed centralized coded caching scheme will be introduced. The delivery rate achieved by this scheme will be analyzed and compared with the existing schemes in the literature. We first illustrate the proposed caching scheme on an example. 

\begin{figure}[!t]
\centering
\includegraphics[scale=0.7]{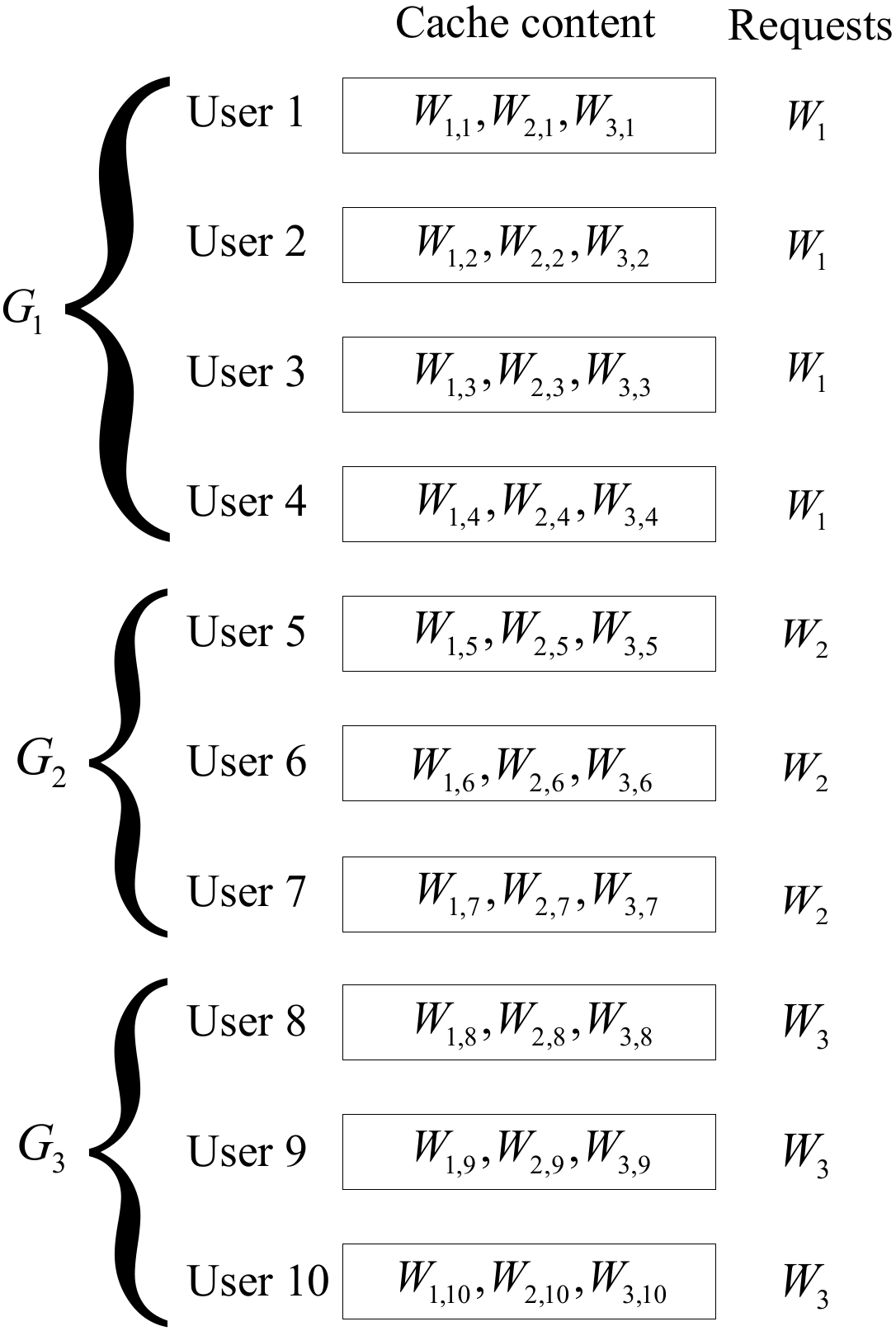}
\caption{Cache contents at the end of the placement phase for the proposed GBC scheme when each of the $K=10$ users demand a single file from among $N=3$ files in the database. The worst-case user demand combination is assumed, and the users are grouped into three groups according to their requests.} 
\label{CacheContent}
\end{figure} 

\theoremstyle{definition}
\newtheorem{exmp}{Example}

\begin{exmp}\label{ExampleCentralized}
In this example, we consider $K = 10$ users, $N = 3$ files, and a cache capacity of $M = 3/10$. Each file $W_i$ is first divided into $10$ non-overlapping subfiles $W_{i,j}$, for $j \in \left[ 1:10 \right]$, and $i \in \left[ 1:3 \right]$, each of length $F/10$ bits. Then one subfile from each file is placed into each user's cache; that is, we have $Z_j = \left( {{W_{1,j}},{W_{2,j}},{W_{3,j}}} \right)$, for $j=1, \ldots, 10$. See Fig. \ref{CacheContent} for the cache contents at the end of the placement phase. 

The worst-case of user demands is when the requested files are as distinct as possible. Without loss of generality, by re-ordering the users, we consider the following user demands:
\begin{equation}\label{DemandsExample} 
{d_j} = 
\begin{cases}
{{1},\quad 1 \le j \le 4},\\
{{2},\quad 5 \le j \le 7},\\
{{3},\quad 8 \le j \le 10}.
\end{cases}
\end{equation}
In the delivery phase, the users are grouped according to their demands. Users that request file $W_i$ from the server constitute group $G_i$, for $i \in \left[ 1:3 \right]$. See Fig. \ref{CacheContent} for the group formation corresponding to the demands in \eqref{DemandsExample}. The delivery phase is divided into two distinct parts, that are designed based on the group structure.

\begin{enumerate}[label=\bfseries Part \arabic*:,align=left]
  \item The first part of the delivery phase is designed to enable each user to retrieve all the subfiles of its demand that have been placed in the caches of users in the same group. As an example, all users $U_1$, ..., $U_4$, i.e., members of group $G_1$, should be able to decode all the subfiles $W_{1,1}$, $W_{1,2}$, $W_{1,3}$, $W_{1,4}$, i.e., the subfiles stored in the cache of users in $G_1$, after receiving the message transmitted in part 1. Accordingly, in our example, in part 1 of the delivery phase the server sends the following coded subfiles over the shared link.  ${W_{1,1}} \oplus {W_{1,2}}$, ${W_{1,2}} \oplus {W_{1,3}}$, ${W_{1,3}} \oplus {W_{1,4}}$, ${W_{2,5}} \oplus {W_{2,6}}$, ${W_{2,6}} \oplus {W_{2,7}}$, ${W_{3,8}} \oplus {W_{3,9}}$, ${W_{3,9}} \oplus {W_{3,10}}$.
  \item The purpose of part 2 is to make sure that, each user can retrieve all the subfiles of its desired file, which have been placed in the cache of users in other groups. Hence, the server transmits the following coded subfiles in the second part of the delivery phase. ${W_{1,5}} \oplus {W_{1,6}}$, ${W_{1,6}} \oplus {W_{1,7}}$, ${W_{2,1}} \oplus {W_{2,2}}$, ${W_{2,2}} \oplus {W_{2,3}}$, ${W_{2,3}} \oplus {W_{2,4}}$, ${W_{1,7}} \oplus {W_{2,4}}$, ${W_{1,8}} \oplus {W_{1,9}}$, ${W_{1,9}} \oplus {W_{1,10}}$, ${W_{3,1}} \oplus {W_{3,2}}$, ${W_{3,2}} \oplus {W_{3,3}}$, ${W_{3,3}} \oplus {W_{3,4}}$, ${W_{1,10}} \oplus {W_{3,4}}$, ${W_{2,8}} \oplus {W_{2,9}}$, ${W_{2,9}} \oplus {W_{2,10}}$, ${W_{3,5}} \oplus {W_{3,6}}$, ${W_{3,6}} \oplus {W_{3,7}}$, ${W_{3,7}} \oplus {W_{2,10}}$.    
\end{enumerate} 
It can be easily verified that, together with the contents placed locally, user $U_k$ can decode its requested file $W_{d_k}$, $\forall k \in \left[ 1:10 \right]$, from the message transmitted in the delivery phase. As a result, by delivering a total of $12F/5$ bits, which corresponds to a delivery rate of $2.4$, all user demands are satisfied. The delivery rate of the best achievable scheme in the literature for cache capacity of $M = 3/10$ can be evaluated from \eqref{BestDeliveryRatePoint}, and it is given by $R_b^{\rm{C}}\left( 3/10 \right)=2.43$.   
\qed
\end{exmp}   

\subsection{Group-Based Centralized Coded Caching (GBC)}\label{GBCscheme}
Here we generalize the ideas introduced above in the example, and introduce the group-based coded caching (GBC) scheme for any $N$ and $K$ values. We consider a cache capacity of $M = N/K$, i.e., the aggregate size of the cache memories distributed across the network is equivalent to the total size of the database.   

Note that, the server has no information about the user demands during the placement phase. Therefore, to satisfy all demand combinations efficiently, and to reduce the dependency of the delivery rate on the particular user demands as much as possible, we retain a symmetry among the subfiles of each file cached at each user. We employ the same placement phase proposed in \cite{MaddahAliCentralized} for $M = N/K$, in which each file $W_i$ is divided into $K$ non-overlapping subfiles $W_{i,j}$, for $i \in \left[ 1:N \right]$ and $j \in \left[ 1:K \right]$, each of the same size $F/K$ bits, and the cache contents of user $U_j$ is given by $Z_j = \left( {{W_{1,j}},{W_{2,j}}, \ldots ,{W_{N,j}}} \right)$, $\forall j \in \left[ 1:K \right]$. It is easy to see that the cache capacity constraint is satisfied. 

Without loss of generality, by re-ordering the users, it is assumed that the first $K_1$ users, referred to as group $G_1$, have the same request $W_1$, the next $K_2$ users, i.e., group $G_2$, request the same file $W_2$, and so on so forth. Hence, there are $K_i \ge 0$ users in each group $G_i$, each with the same request $W_i$, for $i \in \left[ 1:N \right]$. We will see that the delivery rate for the proposed GBC scheme is independent of the values $K_1, K_2, ..., K_N$. An example of grouping users is illustrated in Fig. \ref{CacheContent}, where $K_1=4$, and $K_2=K_3=3$. We define the following variable which will help simplify some of the expressions: 
\begin{align}\label{S_sumK}
S_i \buildrel \Delta \over = \sum\limits_{l = 1}^i {{K_l}}.
\end{align}
The coded delivery phase of the GBC scheme, for the user demands described above, is presented in Algorithm \ref{DeliveryCentralized}.

There are two parts in the delivery phase just as in Example \ref{ExampleCentralized}. Having received the contents delivered in the first part of Algorithm \ref{DeliveryCentralized}, each user can obtain the missing subfiles of its requested file, which are in the cache of users in the same group; that is, each user $U_k$ in group $G_i$ requesting file $W_i$ has access to subfile $W_{i,k}$, and can decode all subfiles $W_{i,l}$, $\forall l \in \left\{ S_{i-1}+1:S_i \right\}$, after receiving the contents delivered in line 3 of Algorithm \ref{DeliveryCentralized}, for $i=1, ..., N$ and $k = S_{i-1}+1, ..., S_i$. With the contents delivered in the second part, each user can decode the subfiles of its requested file which are in the cache of users in other groups. Note that, all users in group $G_i$ demanding file $W_i$ have decoded subfile $W_{j,S_i}$, and they can obtain all subfiles $W_{i,l}$, $\forall l \in \left\{ S_{j-1}+1:S_j \right\}$, i.e., subfiles of file $W_i$ having been cached by users in group $G_j$, after receiving $\mathop  \bigcup \limits_{k = S_{j - 1} + 1}^{S_j - 1} \left( {{W_{i,k}} \oplus {W_{i,k + 1}}} \right)$ and $W_{i, S_j} \oplus W_{j, S_i}$, for $i=1, ..., N-1$ and $j=i+1, ..., N$. Similarly, all users in group $G_j$ can decode all subfiles of their requested file $W_j$, which are in the cache of users in group $G_i$ by receiving $\mathop  \bigcup \limits_{k = S_{i - 1} + 1}^{S_i - 1} \left( {W_{j,k}\oplus W_{j,k+1}} \right)$ and $W_{i, S_j} \oplus W_{j, S_i}$, for $i=1, ..., N-1$ and $j=i+1, ..., N$. In this way, at the end of the proposed delivery phase, the users can recover all the bits of their requested files.  

\begin{algorithm}
\caption{Coded Delivery Phase of the GBC Scheme}
\label{DeliveryCentralized}
\begin{algorithmic}[1]
\Statex
\State{\textbf{Part 1}: Exchanging contents between users in the same group}{}
\For {$i = 1, \ldots, N$}
\State {server delivers $\left( \mathop  \bigcup \limits_{k = S_{i - 1} + 1}^{S_{i}-1} \left( {{W_{i,k}} \oplus {W_{i,k + 1}}} \right) \right)$.}
\EndFor
\Statex
\State{\textbf{Part 2}: Exchanging contents between users in different groups}{}
\For {$i = 1, \ldots, N-1$}
\For {$j = i+1, \ldots, N$}
\State {server delivers $\left( \mathop  \bigcup \limits_{k = S_{j - 1} + 1}^{S_j - 1} \left( {{W_{i,k}} \oplus {W_{i,k + 1}}} \right), \mathop  \bigcup \limits_{k = S_{i - 1} + 1}^{S_i - 1} \left( {W_{j,k}\oplus W_{j,k+1}} \right), W_{i, S_j} \oplus W_{j, S_i} \right)$.}
\EndFor
\EndFor
\end{algorithmic}
\end{algorithm}


\subsection{Delivery Rate Analysis}\label{AnalysisGBC}
Note that, the delivery rate of the GBC scheme, $R^{\mathrm{C}}_{\rm{GBC}}$, should be evaluated for the worst-case user demands. It can be argued that when $N < K$, the worst-case user demands is when there is at least one user requesting each file, i.e., $K_i > 0$, $\forall i \in \left[ 1:N \right]$. On the other hand, when $N \ge K$, without loss of generality, by re-ordering the users, the worst-case user demands is assumed to happen when $K_i = 1$, if $i \in \left[ 1:K \right]$; and $K_i = 0$, otherwise. 

When $N \ge K$, considering the worst-case user demands, the server transmits the contents $\left( \bigcup \limits_{i = 1}^{K - 1} \bigcup \limits_{j = i + 1}^K \left( W_{i,{S_j}} \oplus W_{j,{S_i}} \right) \right)$ using Algorithm \ref{DeliveryCentralized}, which are similar to the contents delivered using the delivery phase proposed in \cite[Algorithm 1]{MaddahAliCentralized} for a cache capacity of $M = N/K$. Thus, when $N \ge K$, the GBC scheme achieves the same delivery rate as the MAN scheme for $M = N/K$, i.e., $R_{\rm{GBC}}^{\rm{C}}(N/K) = R^{\rm{C}}_{\rm{MAN}}(N/K)$.

For $K>N$, the delivery rate of the GBC scheme is stated in the next theorem, whose proof can be found in Appendix \ref{ProofFirstTheorem}. We then compare the result with the best achievable scheme in the literature. 
\begin{theorem}\label{TheoremGBC}
In a centralized coded caching system with $N$ files, each of size $F$ bits, $K$ users, each equipped with a cache of capacity $MF$ bits, where $M = N/K$, the following delivery rate is achievable by the proposed GBC scheme, if $N < K$:
\begin{equation}\label{OurDeliveryRatePointTheorem} {R_{\mathrm{GBC}}^{\rm{C}}}\left( {\frac{N}{K}} \right) = N - \frac{{N\left( {N + 1} \right)}}{{2K}}.
\end{equation}
\end{theorem}

\begin{remark}
It can be seen that the delivery rate of the GBC scheme depends only on $N$ and $K$, and is independent of the values of $K_1,..., K_{N}$, which implies the popularity of the files. The more distinct the files requested by the users, the higher the required delivery rate. Accordingly, the delivery rate given in \eqref{OurDeliveryRatePointTheorem} is obtained for the worst-case user demand combination, such that each file is requested by at least one user.  
\end{remark}

\subsection{Comparison with the State-of-the-Art}
In this subsection, we compare the performance of the GBC scheme with other caching schemes in the literature both analytically and numerically. We first show that the delivery rate achieved for a cache capacity of $M = N/K$ in \eqref{OurDeliveryRatePointTheorem} is the best known result as long as $N < K$. We then extend this improvement to a wider range of cache capacities through memory-sharing.

\begin{corollary}  
The achievable delivery rate of the GBC scheme for a cache capacity of $M = N/K$ given in \eqref{OurDeliveryRatePointTheorem} improves upon the best known delivery rate in the literature, $R^{\rm{C}}_b(N/K)$, when $K > N \ge 3$. 
\end{corollary}

\begin{proof}
We compare the delivery rate of the GBC scheme with the best achievable scheme in the literature described in Section \ref{SystemModel} for both cases $\left( {N,K} \right) \in \zeta$ and $\left( {N,K} \right) \notin \zeta$. For the first case, the slope of the delivery rate for $1/K \le M \le (N-1)/K$, which is achieved by memory-sharing between the CFL and AG schemes for points $M = 1/K$ and $M = (N-1)/K$, respectively, is $\frac{(K - 2)(K - 2N)}{2(N - 2)}$. On the other hand, the slope of the delivery rate for $1/K \le M \le N/K$, achieved through memory-sharing between the CFL and GBC schemes for points $M = 1/K$ and $M = N/K$, respectively, is $- N/2$. Since for $\left( {N,K} \right) \in \zeta$, we have $N \ge 4$, and $N+2 \le K \le 3N/2$, it can be easily verified that the latter lies below the former; that is, by exploiting the GBC scheme rather than the AG scheme, a lower delivery rate can be achieved at $M = N/K$. 

For a comparison of $R^{\rm{C}}_{\mathrm{GBC}}\left( N/K \right)$ and $R_b^{\rm{C}}(N/K)$ in the second case, we need to evaluate $R_b^{\rm{C}}\left( N/K \right)$ given in \eqref{BestDeliveryRatePoint}, which is not straightforward. Instead, in Appendix \ref{CentralizedImprovement}, we prove that for $K > N \ge 3$, ${R^{\rm{C}}_{\mathrm{GBC}}}\left( N/K \right) \le f\left( N,K,t \right)$, $\forall t \in \left\{ 1:K \right\}$. Since $R_b^{\rm{C}}\left( N/K \right) = f\left( N,K,t^* \right)$, where $t^* \in \left\{ 1:K \right\}$, it can be concluded that ${R^{\rm{C}}_{\mathrm{GBC}}}\left( N/K \right) \le {R^{\rm{C}}_b}\left( N/K \right)$. This completes the proof. 
\end{proof}

\begin{remark}
We observe through simulations that, apart from the case $N=2$, the centralized scheme proposed in \cite{KaiWanUncodedCaching}, referred to as WTP, does not improve upon time-sharing between the MAN and CFL schemes. Nevertheless, below we explicitly demonstrate the superiority of the GBC scheme over WTP for a cache capacity of $M = N/K$, when $K > N \ge 2$. Let $R^{\rm{C}}_{\mathrm{WTP}}\left( M \right)$ denote the delivery rate-cache capacity trade-off of the WTP scheme in the centralized setting. According to \cite[Theorem 1]{KaiWanUncodedCaching}, we have
\begin{equation}\label{WTPDeliveryRate} R_{\mathrm{WTP}}^{\rm{C}}\left(\frac{N}{K} \right) \ge \min \left( {R_{{\mathrm{co}}_1}\left( {\frac{N}{K}} \right),R_{{\mathrm{co}}_2}\left( {\frac{N}{K}} \right)} \right),
\end{equation}  
where
\begin{subequations}
\label{WTPDeliveryRateco}
\begin{align}
{R_{{\rm{c}}{{\rm{o}}_1}}}\left( M \right) &= N - M - \frac{{M\left( {N - 1} \right)K\left( {N - M} \right)}}{{{N^2}\left( {K - 1} \right)}},\\
{R_{{\rm{c}}{{\rm{o}}_2}}}\left( M \right) &= \frac{{K\left( {N - M} \right)}}{{N + KM}}.
\end{align}
\end{subequations}
When $K > N \ge 2$, it can be easily verified that 
\begin{subequations}
\label{CompareOursWTPDeliveryRateco}
\begin{align}
{R^{\rm{C}}_{{\rm{GBC}}}}\left( {\frac{N}{K}} \right) &\le {R_{{\rm{c}}{{\rm{o}}_1}}}\left( {\frac{N}{K}} \right) = N + \frac{1}{K} - \frac{{2N}}{K},\\
{R^{\rm{C}}_{{\rm{GBC}}}}\left( {\frac{N}{K}} \right) &\le {R_{{\rm{c}}{{\rm{o}}_2}}}\left( {\frac{N}{K}} \right) = \frac{{K - 1}}{2},
\end{align}
\end{subequations}
which concludes that 
\begin{equation}\label{CompareOursWTPDeliveryRate} {R^{\rm{C}}_{{\rm{GBC}}}}\left( {\frac{N}{K}} \right) \le {R^{\rm{C}}_{{\rm{WTP}}}}\left( {\frac{N}{K}} \right).
\end{equation}  
\end{remark} 

The improvement obtained by using the GBC scheme for $M = N/K$, when $N < K$, can be extended to any cache capacities $1/K < M < \hat tN/K$, where $\hat t \buildrel \Delta \over = \max \left\{ {2,{t^*}} \right\}$, through memory-sharing between the CFL scheme for $M = 1/K$, the GBC scheme for $M = N/K$, and the MAN scheme for cache capacity $M = \hat tN/K$. Note that, when $t^* = 1$, the improvement can be extended to the interval $1/K < M < 2N/K$, while for $t^* \ge 2$, the superiority of the GBC scheme holds for cache capacities $1/K < M < t^*N/K$.

\begin{corollary}  
For $N$ files and $K$ users, each equipped with a cache of normalized capacity $M$, the following delivery rate-cache capacity trade-off is achievable in the centralized setting if $K>N$:
\begin{align}\label{OurDeliveryRateIntervalCorollary} 
R^{\rm{C}}_{\mathrm{GBC}}\left( M \right) = 
\begin{cases}
N \left( 1 - \frac{M}{2} - \frac{1}{2K} \right), \qquad \qquad \qquad \quad \;  \mbox{ if }1/K \le M \le N/K,\\
\frac{K-\hat{t}}{\hat{t}^2-1}\left(\frac{KM}{N} -1\right) + \frac{K-N}{\hat{t}-1} \left( \hat{t}N/K-M \right), \mbox{ if } N/K \le M \le \hat{t}N/K,
\end{cases}
\end{align}
where $\hat t \buildrel \Delta \over = \max \left\{ {2,{t^*}} \right\}$, and $t^*$ is determined as in \eqref{optimumt}.
\end{corollary}

\begin{figure}[!t]
\centering
\includegraphics[scale=0.7]{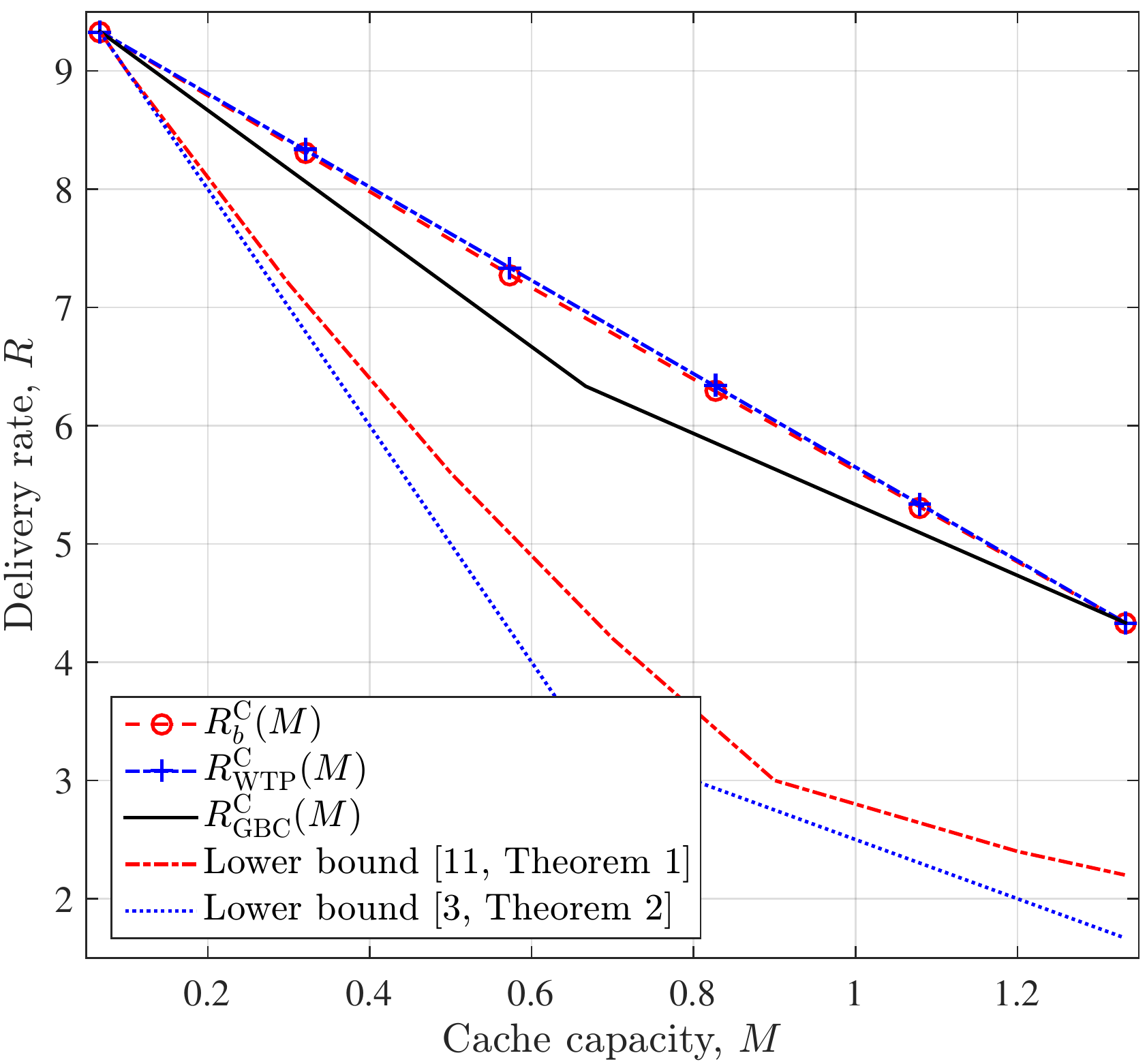}
\caption{Delivery rate cache capacity trade-off of the proposed GBC scheme compared with the existing schemes in the literature for $N=10$ and $K=15$. This setting corresponds to Case 1 in Section \ref{SystemModel}. Since $K = 3N/2$, for $1/K \le M \le 2N/K$, $R_{\rm{b}}^{\rm{C}}$ is achieved by memory-sharing between CFL, AG, and MAN schemes.} 
\label{N10K15}
\end{figure}

Next, we illustrate the improvements in the delivery rate offered by the proposed GBC scheme through numerical simulations. The delivery rate of known centralized caching schemes as a function of the cache capacity, $M$, are compared in Fig. \ref{N10K15}, for $N=10$ files and $K=15$ users. Note that, these $N$ and $K$ values correspond to Case 1 for the best achievable scheme in the literature, i.e., $\left( {N,K} \right) \in \zeta$. The best known achievable delivery rate in the literature for Case 1 is achieved by memory-sharing between the CFL, AG, and MAN schemes. Since we have $K = 3N/2$, the schemes are compared over the range $1/K \le M \le 2N/K$, where GBC provides an improvement. We also include the centralized WTP scheme, whose performance is slightly worse than that achieved by memory-sharing between the CFL, AG, and MAN schemes in this scenario. The information theory and cut-set lower bounds investigated in \cite[Theorem 1]{SenguptaCaching} and \cite[Theorem 2]{MaddahAliCentralized}, respectively, are also included in this figure for comparison. We observe a significant improvement in the delivery rate offered by the GBC scheme compared to all other schemes in the literature for all cache capacity values $1/K < M < 2N/K$. Despite the improvement upon other known schemes, there is still a gap between $R_{\mathrm{GBC}}$ and the theoretical lower bound; although this gap might as well be due to the looseness of the known lower bounds.

In Fig. \ref{N50K130}, we compare the delivery rate of the GBC scheme with the existing schemes for $N=50$ and $K=130$. This scenario corresponds to Case 2 of the best achievable schemes in the literature, i.e., $\left( {N,K} \right) \notin \zeta$, whose delivery rate is achieved through memory-sharing between the CFL and MAN schemes. According to \eqref{optimumt}, we have $t^*=4$, and hence, $\hat t = t^* = 4$. The centralized WTP scheme achieves the same performance as memory-sharing between the CFL and MAN schemes in this scenario. We again observe that the GBC scheme achieves a significantly lower delivery rate compared to the known caching schemes over the range of cache capacities $1/K \le M \le 4N/K$. 

\begin{figure}[!t]
\centering
\includegraphics[scale=0.7]{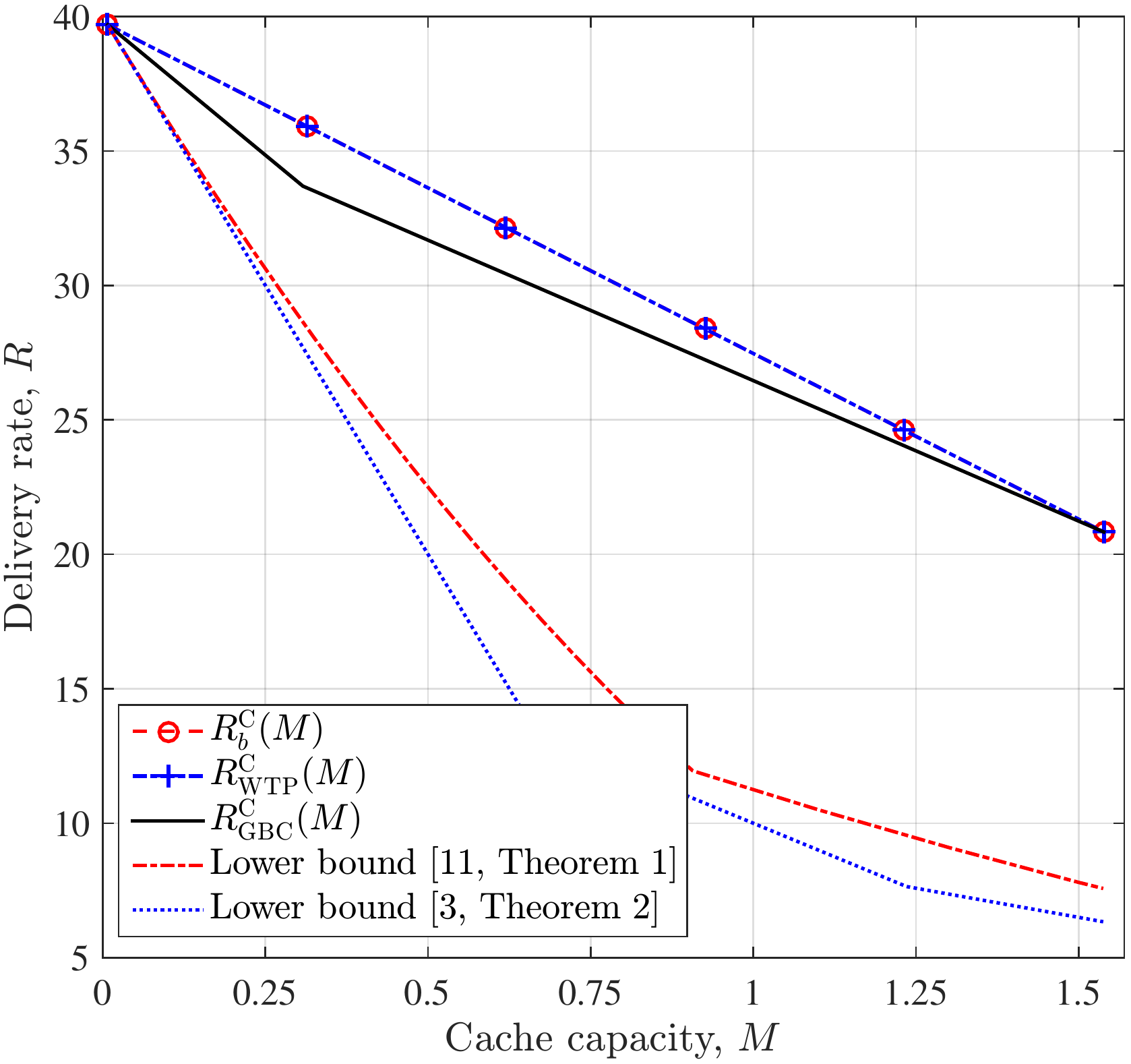}
\caption{Delivery rate cache capacity trade-off of the proposed GBC scheme compared with the best achievable scheme in the literature for $N = 50$ and $K = 130$.  This setting corresponds to Case 2 in Section \ref{SystemModel}. $R_{\rm{b}}^{\rm{C}}$ is achieved by memory-sharing between the CFL and MAN schemes for $1/K \le M \le 4N/K$.} 
\label{N50K130}
\end{figure} 

In Fig. \ref{NConstantKVariable}, the delivery rate of the proposed GBC scheme for cache capacity $M = N/K$, i.e., $R^{\rm{C}}_{\rm{GBC}}( N/K )$ given in \eqref{OurDeliveryRatePointTheorem}, is compared with the best achievable delivery rate in the literature for the same cache size, when the number of files is $N=100$, and the number of users varies from $K = 200$ to $K = 1000$. In this scenario, since $K \ge 2N$, the best achievable rate in the literature, $R_b^{\rm{C}}( N/K )$, is determined by \eqref{BestDeliveryRatePoint}. The WTP scheme is not included as it again achieves the same delivery rate as $R_b^{\rm{C}}( N/K )$. We observe that the GBC provides a significant gain in the delivery rate for the whole range of $K$ values; however, the improvement is more pronounced for smaller number of users, $K \sim 200$. For $K=200$ users in the system, the GBC scheme provides a $9.75\%$ reduction in the delivery rate. We also observe that the gap between the delivery rate of the GBC scheme and the lower bound decreases with increasing number of users in the system.      

\begin{figure}[!t]
\centering
\includegraphics[scale=0.7]{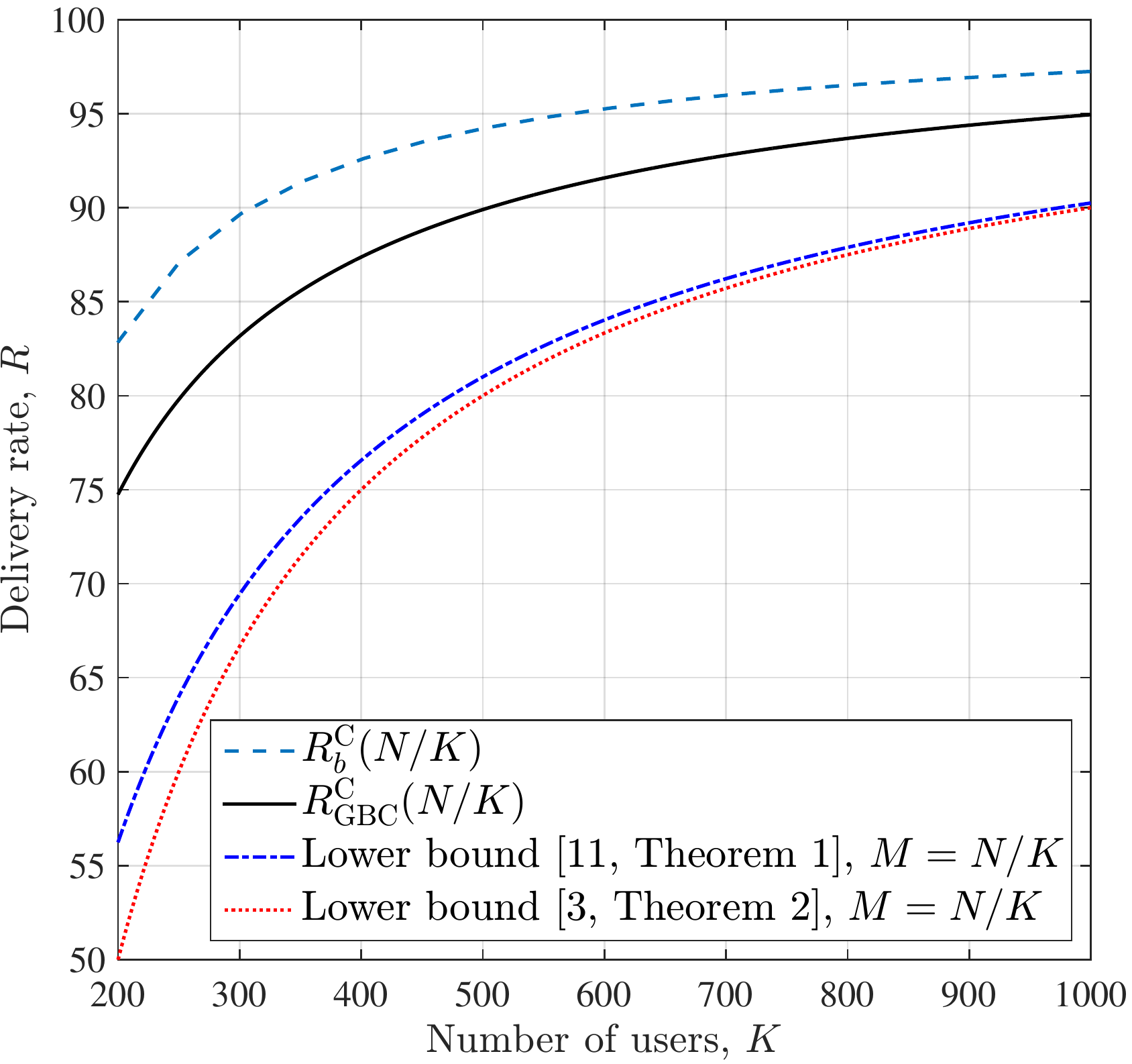}
\caption{Delivery rate of the GBC scheme for cache capacity $M = N/K$ as a function of $K$, for $K \in \left[ 200:1000 \right]$, and $N=100$. The best achievable delivery rate in the literature is determined by \eqref{BestDeliveryRatePoint} for all values of $K$.} 
\label{NConstantKVariable}
\end{figure}

\section{Decentralized Coded Caching}\label{DecentralizedScheme}
In this section, we consider a decentralized caching system, in which neither the number nor the identity of the users that participate in the delivery phase are known during the placement phase \cite{MaddahAliDecentralized}. Accordingly, in the decentralized setting, we cannot coordinate the cache contents among users in order to maximize the multicasting opportunities during the delivery phase. We apply the group-based caching ideas we have developed for the centralized scenario to decentralized caching. The corresponding caching scheme is called the group-based decentralized coded caching (GBD). 

To simplify the notation, for $k=1, ..., K$, we define $W_{d_k,V}$ as the bits of file $W_{d_k}$, the file requested by user $U_k$, which have been placed exclusively in the caches of the users in set $V$, where $V \subset \left[ {1:K} \right]$. We also note that, all XOR operations used in this section are assumed to be zero-padded such that they all have the same length as the longest of their arguments. The GBD scheme is first demonstrated on a simple example. 

\begin{exmp}
Consider a decentralized coded caching system with $K=5$ users and $N=3$ files. Due to lack of coordination in the placement phase, similarly to the scheme proposed in \cite{MaddahAliDecentralized}, $MF/3$ random bits of each file are cached independently at each user during the placement phase. During the delivery phase, without loss of generality, by re-ordering the users, the following worst-case user demands are considered:
\begin{equation}\label{DemandsExample2} 
{d_k} =
\begin{cases} 
{{1},\quad 1 \le k \le 2,}\\
{{2},\quad 3 \le k \le 4,}\\
{{3},\quad k = 5.}
\end{cases}
\end{equation}
The contents are sent by the server in three parts during the delivery phase to satisfy the above demand combination. Below, we explain each part in detail. 

\begin{enumerate}[label=\bfseries Part \arabic*:,align=left]
\item In the first part, the bits of each requested file $W_i$, which are not cached anywhere, are delivered directly by the server, for $i \in \left[ 1:3 \right]$. Therefore, the server delivers the following bits in the first part of the delivery phase: ${W_{1,\left\{ \emptyset \right\} }}$, ${W_{2,\left\{ \emptyset \right\} }}$, ${W_{3,\left\{ \emptyset \right\} }}$.

\item In part 2, the bits of file $W_{d_k}$ requested by user $U_k$, which are only in the cache of user $U_j$ are delivered, for $k,j \in \left[ 1:5 \right]$ and $k \ne j$. In other words, the purpose of this part is to enable each user to decode the bits of its requested file which have been placed in the cache of only one other user. Accordingly, the delivery phase of the GBC scheme can be applied in this part. For the example under consideration, the following bits constitute the second part of the delivery phase: ${W_{1,\left\{1\right\}}} \oplus {W_{1,\left\{2\right\}}}$, ${W_{2,\left\{3\right\}}} \oplus {W_{2,\left\{4\right\}}}$, ${W_{1,\left\{3\right\}}} \oplus {W_{1,\left\{4\right\}}}$, ${W_{2,\left\{1\right\}}} \oplus {W_{2,\left\{2\right\}}}$, ${W_{1,\left\{4\right\}}} \oplus {W_{2,\left\{2\right\}}}$, ${W_{3,\left\{1\right\}}} \oplus {W_{3,\left\{2\right\}}}$, ${W_{1,\left\{5\right\}}} \oplus {W_{3,\left\{1\right\}}}$, ${W_{3,\left\{3\right\}}} \oplus {W_{3,\left\{4\right\}}}$, ${W_{2,\left\{5\right\}}} \oplus {W_{3,\left\{4\right\}}}$. 

\item The bits of the file requested by each user which are in the cache of more than one other user are delivered in the third part by exploiting the procedure proposed in \cite[Algorithm 1]{MaddahAliDecentralized}. The server delivers the following in the third and last part of the delivery phase: ${W_{1,\left\{ {2,3} \right\}}} \oplus {W_{1,\left\{ {1,3} \right\}}} \oplus {W_{2,\left\{ {1,2} \right\}}}$, ${W_{1,\left\{ {2,4} \right\}}} \oplus {W_{1,\left\{ {1,4} \right\}}} \oplus {W_{2,\left\{ {1,2} \right\}}}$, ${W_{1,\left\{ {2,5} \right\}}} \oplus {W_{1,\left\{ {1,5} \right\}}} \oplus {W_{3,\left\{ {1,2} \right\}}}$, ${W_{1,\left\{ {3,4} \right\}}} \oplus {W_{2,\left\{ {1,4} \right\}}} \oplus {W_{2,\left\{ {1,3} \right\}}}$, ${W_{1,\left\{ {3,5} \right\}}} \oplus {W_{2,\left\{ {1,5} \right\}}} \oplus {W_{3,\left\{ {1,3} \right\}}}$, ${W_{1,\left\{ {4,5} \right\}}} \oplus {W_{2,\left\{ {1,5} \right\}}} \oplus {W_{3,\left\{ {1,4} \right\}}}$, ${W_{1,\left\{ {2,3,4} \right\}}} \oplus {W_{1,\left\{ {1,3,4} \right\}}} \oplus {W_{2,\left\{ {1,2,4} \right\}}} \oplus {W_{2,\left\{ {1,2,3} \right\}}}$, ${W_{1,\left\{ {2,3,5} \right\}}} \oplus {W_{1,\left\{ {1,3,5} \right\}}} \oplus {W_{2,\left\{ {1,2,5} \right\}}} \oplus {W_{3,\left\{ {1,2,3} \right\}}}$, ${W_{1,\left\{ {2,4,5} \right\}}} \oplus {W_{1,\left\{ {1,4,5} \right\}}} \oplus {W_{2,\left\{ {1,2,5} \right\}}} \oplus {W_{3,\left\{ {1,2,4} \right\}}}$, ${W_{1,\left\{ {3,4,5} \right\}}} \oplus {W_{2,\left\{ {1,4,5} \right\}}} \oplus {W_{2,\left\{ {1,3,5} \right\}}} \oplus {W_{3,\left\{ {1,3,4} \right\}}}$, ${W_{1,\left\{ {2,3,4,5} \right\}}} \oplus {W_{1,\left\{ {1,3,4,5} \right\}}} \oplus {W_{2,\left\{ {1,2,4,5} \right\}}} \oplus {W_{2,\left\{ {1,2,3,5} \right\}}} \oplus {W_{3,\left\{ {1,2,3,4} \right\}}}$, ${W_{1,\left\{ {3,4} \right\}}} \oplus {W_{2,\left\{ {2,4} \right\}}} \oplus {W_{2,\left\{ {2,3} \right\}}}$, ${W_{1,\left\{ {3,5} \right\}}} \oplus {W_{2,\left\{ {2,5} \right\}}} \oplus {W_{3,\left\{ {2,3} \right\}}}$, ${W_{1,\left\{ {4,5} \right\}}} \oplus {W_{2,\left\{ {2,5} \right\}}} \oplus {W_{3,\left\{ {2,4} \right\}}}$, ${W_{1,\left\{ {3,4,5} \right\}}} \oplus {W_{2,\left\{ {2,4,5} \right\}}} \oplus {W_{2,\left\{ {2,3,5} \right\}}} \oplus {W_{3,\left\{ {2,3,4} \right\}}}$, ${W_{2,\left\{ {4,5} \right\}}} \oplus {W_{2,\left\{ {3,5} \right\}}} \oplus {W_{3,\left\{ {3,4} \right\}}}$.        
\end{enumerate} 

In this case, using the law of large numbers, the size of $W_{i,V}$, for $F$ large enough, can be approximated as
\begin{equation}\label{SizePortion} \left| {{W_{i,V}}} \right| \approx {\left( {\frac{M}{3}} \right)^{\left| V \right|}}{\left( {1 - \frac{M}{3}} \right)^{5 - \left| V \right|}}F,
\end{equation} 
for any set $V \subset \left[ {1:5} \right]$. For example, for a cache capacity of $M=1$, the delivery rate of the proposed coded caching scheme is $1.407$, while the schemes proposed in \cite{MaddahAliDecentralized} and \cite{KaiWanUncodedCaching}, can achieve the delivery rates $1.737$ and $1.473$, respectively.
\qed
\end{exmp}   

\subsection{Group-Based Decentralized Coded Coding (GBD) Scheme}
Here, we generalize the proposed decentralized group-based coded caching (GBD) scheme. To simplify the notation, for $n,m \in \mathcal Z$, and $V \subset \left[1:K \right]$, we use the notation $\left\{ {V,n} \right\}\backslash \left\{ m \right\}$ to represent the set of integers $V \cup \{n\} \backslash \{m\}$. 

In the placement phase, each user caches a random subset of $MF/N$ bits of each file independently. Since there are $N$ files, each of length $F$ bits, this placement phase satisfies the memory constraint.

\begin{algorithm}
\caption{Coded Delivery Phase of the GBD Scheme}
\label{DeliveryDecentralized}
\begin{algorithmic}[1]
\Statex
\Procedure {Delivery-Coded}{}
\State{\textbf{Part 1}: Delivering bits that are not in the cache of any user}
\For {$i = 1, 2, \ldots, N$}
\State{server delivers $\left( W_{{d_{S_{i-1} + 1}},\left\{ \emptyset \right\}}  \right)$}
\EndFor
\Statex
\State{\textbf{Part 2}: Delivering bits that are in the cache of only one user}
\State{server delivers $\left( {\bigcup\limits_{i = 1}^N {\bigcup\limits_{k = {S_{i - 1}} + 1}^{{S_i} - 1} {\left( {{W_{i,\left\{ k \right\}}} \oplus {W_{i,\left\{ k + 1 \right\}}}} \right)} } } \right)$.}
\State{server delivers $\bigcup\limits_{i = 1}^{N - 1} \bigcup\limits_{j = i + 1}^N \left( \bigcup\limits_{k = {S_{j - 1}} + 1}^{{S_j} - 1} {\left( {{W_{i,\left\{ k \right\}}} \oplus {W_{i,\left\{ {k + 1} \right\}}}} \right)} ,\qquad \qquad \qquad \qquad \qquad \qquad \qquad \right.$ $ \left. \qquad \qquad \qquad \qquad \qquad \qquad \qquad \qquad \bigcup\limits_{k = {S_{i - 1}} + 1}^{{S_i} - 1} {\left( {{W_{j,\left\{ k \right\}}} \oplus {W_{j,\left\{ {k + 1} \right\}}}} \right)} , {W_{i,\left\{ {{S_j}} \right\}}} \oplus W_{j,\left\{ {{S_i}} \right\}} \right) $.}
\Statex
\State{\textbf{Part 3}: Delivering bits that are in the cache of more than one user}
\For{$i = 1, 2,  \ldots, K - 2$ }
\For{$j = 2, 3, \ldots, K - i$ }
\For{$V \subset \left[ {i + 1:K} \right]: \left| V \right| = j$ }
\State server delivers $\left( \left( \mathop  \oplus \limits_{v \in V} W_{d_{v},\left\{ {V,i} \right\}\backslash \left\{v\right\}} \right) \oplus W_{d_i, V } \right)$
\EndFor
\EndFor
\EndFor
\EndProcedure
\Statex
\Procedure {Delivery-Random}{}
\For{$i = 1, 2, \ldots, N$}
\State {server delivers enough random linear combination of the bits of the file $W_i$ to the users requesting it in order to decode it}
\EndFor
\EndProcedure
\end{algorithmic}
\end{algorithm}

Similarly to Section \ref{GBCscheme}, without loss of generality, we can re-order the users and re-label the files such that the first $K_1$ users, referred to as group $G_1$, have the same request $W_1$, the next $K_2$ users, group $G_2$, request $W_2$, and so on so forth. As a result, in the delivery phase, we have
\begin{equation}\label{DemandsDecentralized} d_k = i,\quad \mbox{for $i=1, ..., N$, and $k=S_{i-1} + 1, ..., S_i$},
\end{equation} 
where $S_i$ is as defined in (\ref{S_sumK}).

There are two different procedures for the delivery phase, called DELIVERY-CODED and DELIVERY-RANDOM, presented in Algorithm \ref{DeliveryDecentralized}. The server follows either of the two, whichever achieves a smaller delivery rate.

Let us start with the DELIVERY-CODED procedure of Algorithm \ref{DeliveryDecentralized}, in which the contents are delivered in three distinct parts, as explained in the example above. The main idea behind the coded delivery phase is to deliver each user the missing bits of its requested file, that have been cached by $i$ user(s), $\forall i \in \left[ 0:K-1 \right]$. 

In the first part, the bits of each requested file that are not in the cache of any user are directly delivered by the server. Each transmitted content is destined for all the users in a separate group, which have the same request. 

In the second part, the bits of each requested file that have been cached by only one user are served to the users requesting the file by utilizing the GBC scheme developed for the centralized scenario. 

Each user $U_k$ in group $G_i$ requests $W_i$ and has already cached $W_{i,\left\{ k \right\}}$ for $i \in \left[1:N \right]$ and $k \in \left[ S_{i-1}+1:S_i \right]$. Having received the bits delivered in line 7 of Algorithm \ref{DeliveryDecentralized}, $U_k$ can decode all bits $W_{i,\left\{ l \right\}}$, $\forall l \in \left[ S_{i-1}+1:S_i \right]$. The users also receive the missing bits of their requested files having been cached by a user in a different group; that is, by receiving $\bigcup\limits_{k = S_{i - 1} + 1}^{S_i - 1} \left( W_{j,\left\{k\right\}} \oplus W_{j,\left\{k + 1\right\}} \right)$, $\bigcup\limits_{k = {S_{j - 1}} + 1}^{{S_j} - 1} {\left( {{W_{i,\left\{k\right\}}} \oplus {W_{i,\left\{k + 1\right\}}}} \right)}$, and ${W_{i,\left\{S_j\right\}}} \oplus {W_{j,\left\{S_i\right\}}}$, each user in groups $G_i$ and $G_j$ can decode the bits of its request which have been placed in the cache of users in the other group, for $i=1, ..., N-1$ and $j=i+1, ..., N$.

In the last part, the same procedure as the one proposed in \cite{MaddahAliDecentralized} is performed for the missing bits of each file that have been cached by more than one user. Hence, following the DELIVERY-CODED procedure presented in Algorithm \ref{DeliveryDecentralized}, each user recovers all the bits of its desired file.

The second delivery procedure, DELIVERY-RANDOM, is as presented in In the last part, the same procedure as the one proposed in \cite{MaddahAliDecentralized}, and the server delivers enough random linear combinations of the bits of each requested file targeted for the users in the same group requesting that file to decode it. 

\subsection{Delivery Rate Analysis}
In the following, we derive an expression for the delivery rate-cache capacity trade-off of the proposed GBD scheme, denoted by $R^{\rm{D}}_{\rm{GBD}}(M)$. All discussions in this section are stated assuming that $M \le N$, and $F$ is large enough. For each randomly chosen bit of each file, the probability of having been cached by each user is $M/N$. Since the contents are cached independently by each user in the placement phase, a random bit of each file is cached exclusively by the users in set $V \subset \left[ 1:K \right]$ (and no user outside this set) with probability $\left( M/N \right)^{\left| V \right|} \left( {1 - M/N} \right)^{K - \left| V \right|}$. 

Similarly to the arguments presented in Section \ref{AnalysisGBC}, when $N < K$, the worst-case user demands correspond to the scenario in which each file is requested by at least one user, i.e., $K_i > 0$, $\forall i \in \left[ {1:N} \right]$. On the other hand, when $N \ge K$, without loss of generality, the worst-case user demands can be assumed as $d_k = k$, $\forall k \in \left[ {1:K} \right]$.

When $N \ge K$, for the worst-case user demands described above, similar conditions as the GBC scheme hold, and the GBD scheme achieves the same delivery rate as the decentralized caching scheme proposed in \cite{MaddahAliDecentralized}, called the decentralized MAN scheme, i.e., $R^{\rm{D}}_{\rm{GBD}}\left( M \right) = R^{\rm{D}}_{\rm{MAN}}\left( M \right)$. 

Next we consider the more interesting $N < K$ case. We start with the first procedure of Algorithm \ref{DeliveryDecentralized}. In part 1, the server delivers $N$ groups of bits, each group corresponding to a different file, which have not been cached by any user. The delivery rate-cache capacity trade-off for this part of Algorithm \ref{DeliveryDecentralized}, $R^{\rm{D}}_{1}(M)$, can be evaluated as  
\begin{equation}\label{OurDeliveryRateDecentralized1} {R^{\rm{D}}_{{1}}}\left( M \right) = N{\left( {1 - \frac{M}{N}} \right)^K}.
\end{equation} 
For the second part, we first need to find the total number of XOR-ed contents delivered by the server, each of which has length $\left( M/N \right) \left( 1-M/N \right)^{K-1}F$ bits (the length of each XOR-ed content is equivalent to the number of bits of a file that have been cached by only one user). Since the delivery phase of the GBC scheme is applied for this part, based on \eqref{OurDeliveryRatePointTheorem}, it can be easily evaluated that $\left( {NK - N(N+1)/2} \right)$ XOR-ed contents are served\footnote{Note that, in the delivery phase of the GBC scheme for $M = N/K$, a total of $\left( NK - N(N+1)/2 \right)$ XOR-ed contents, each of size $F/K$ bits, are delivered, which results in $R^{\rm{C}}_{\rm{GBC}}( N/K )=N-N(N+1)/2K$.}. Thus, the delivery rate-cache capacity trade-off corresponding to the second part of the first procedure of Algorithm \ref{DeliveryDecentralized} is given by
\begin{equation}\label{OurDeliveryRateDecentralized2} {R^{\rm{D}}_{{2}}}\left( M \right) = \left( {NK - \frac{{N\left( {N + 1} \right)}}{2}} \right)\left( {\frac{M}{N}} \right){\left( {1 - \frac{M}{N}} \right)^{K - 1}}.
\end{equation}   

The last part of the proposed delivery phase is equivalent to the first delivery phase procedure proposed in \cite[Algorithm 1]{MaddahAliDecentralized}, with which each user can decode the bits of its requested file, which have been cached by more than one user. Following the same technique as \cite{MaddahAliDecentralized}, the delivery rate corresponding to this part is derived as follows:
\begin{align}\label{OurDeliveryRateDecentralized3} {R^{\rm{D}}_{{3}}}\left( M \right) = & \sum\limits_{i = 1}^{K - 2} {\sum\limits_{j = 2}^{K - i} {\binom{K-i}{j}{{\left( {\frac{M}{N}} \right)}^j}{{\left( {1 - \frac{M}{N}} \right)}^{K - j}}} }  \nonumber\\ 
= & - \left( {K - 2} \right){\left( {1 - \frac{M}{N}} \right)^K} - \frac{1}{2}\left( {K - 2} \right)\left( {K + 1} \right){\left( {\frac{M}{N}} \right)}{\left( {1 - \frac{M}{N}} \right)^{K - 1}} \nonumber\\
&~~~~~~ + \frac{N}{M}\left( {1 - {{\left( {1 - \frac{M}{N}} \right)}^{K - 1}}} \right) - 1.
\end{align}   

The overall delivery rate-cache capacity trade-off for the first procedure of Algorithm \ref{DeliveryDecentralized}, $R^{\rm{D}}_{\rm{GBD}_1}(M)$, is the sum of the delivery rates of all three parts in (\ref{OurDeliveryRateDecentralized1}), (\ref{OurDeliveryRateDecentralized2}), and (\ref{OurDeliveryRateDecentralized3}), and is evaluated as
\begin{align} {R^{\rm{D}}_{{\rm{GBD}_1}}}\left( M \right) &= {R^{\rm{D}}_{{1}}}\left( M \right) + {R^{\rm{D}}_{{2}}}\left( M \right) + {R^{\rm{D}}_{{3}}}\left( M \right) \nonumber \\
    &= \frac{N}{M} - 1 - \left[ {\left( {K - N - 2} \right)\left( {1 + \frac{1}{2}\left( {K - N - 1} \right)\frac{M}{N}} \right) + \frac{N}{M}} \right]{\left( {1 - \frac{M}{N}} \right)^{K - 1}}. \label{OurDeliveryRateDecentralizedFirstProcedure}
\end{align} 

For the worst-case user demands, it is shown in \cite[Appendix A]{MaddahAliDecentralized} that the second delivery procedure achieves the same delivery rate-cache capacity trade-off as the uncoded scheme, given by
\begin{equation}\label{UncodedDeliveryRate} {R^{\rm{D}}_{\rm{GBD}_2}}\left( M \right) = K\left( {1 - \frac{M}{N}} \right)\min \left\{ {1,\frac{N}{K}} \right\}.
\end{equation}

The delivery rate of the proposed GBD scheme is evidently the minimum value of $R^{\rm{D}}_{\rm{GBD}_1}(M)$ and $R^{\rm{D}}_{\rm{GBD}_2}(M)$, which is presented in the following theorem.

\begin{theorem}  
In a decentralized caching system with $K$ users requesting contents from a server with $N$ files in its database, when $N < K$, the following delivery rate-cache capacity trade-off is achievable by the GBD scheme:
\begin{align}\label{DeliveryRateAlgorithm2} 
& R^{\rm{D}}_{\rm{GBD}}\left( M \right) = \left( {1 - \frac{M}{N}} \right)\nonumber\\ 
& \quad \; \times \min \left\{ \frac{N}{M} - \left[ {\left( {K - N - 2} \right)\left( {1 + \frac{1}{2}\left( {K - N - 1} \right)\frac{M}{N}} \right) + \frac{N}{M}} \right]{{\left( {1 - \frac{M}{N}} \right)}^{K - 2}},N \right\}.
\end{align}
\end{theorem}

\begin{remark}
We remark that in the decentralized caching model, it is assumed that each user sends its cache content together with its request to the server at the beginning of the delivery phase. In this way, by knowing the number of popular files in the database, when $N<K$, the server can decide to perform the delivery phase procedure that requires the smallest delivery rate.
\end{remark}

\subsection{Comparison with the State-of-the-Art}
The difference between the proposed GBD scheme and the scheme investigated in \cite[Algorithm 1]{MaddahAliDecentralized} lies in the first procedure of the delivery phase when $N < K$. As a result, to compare the two schemes, the delivery rate of the proposed GBD scheme for the DELIVERY-CODED procedure, i.e., $R_{\rm{GBD}_1}^{\rm{D}}(M)$ given in \eqref{OurDeliveryRateDecentralizedFirstProcedure}, should be compared with the delivery rate of the first delivery phase procedure stated in \cite[Algorithm 1]{MaddahAliDecentralized}, which is given by
\begin{equation}\label{DeliveryRateDecentralizedFirstProcedureMaddahAli} R_{\rm{MAN}_1}^{\rm{D}}\left( M \right) = \left( \frac{N}{M} - 1 \right) \left( 1 - \left( {1 - \frac{M}{N}} \right)^K \right).
\end{equation} 
For $M \le N$, we have ${\left( {1 - M/N} \right)^K} \le {\left( {1 - M/N} \right)^{K - 1}}$. Hence, to show that ${R_{\rm{GBD}_1}^{\rm{D}}}\left( M \right) \le {R_{\rm{MAN}_1}^{\rm{D}}}\left( M \right)$, it suffices to prove that
\begin{align}\label{DeliveryRateDecentralizedFirstProcedureCompare} \left( {K - N - 2} \right)\left( {1 + \frac{1}{2}\left( {K - N - 1} \right)\frac{M}{N}} \right) \ge  - 1, 
\end{align}   
which holds for $N < K$. Therefore, compared to the decentralized coded caching scheme proposed in \cite[Algorithm 1]{MaddahAliDecentralized}, the GBD scheme requires a smaller delivery rate, if $N < K$.

\begin{figure}[!t]
\centering
\includegraphics[scale=0.7]{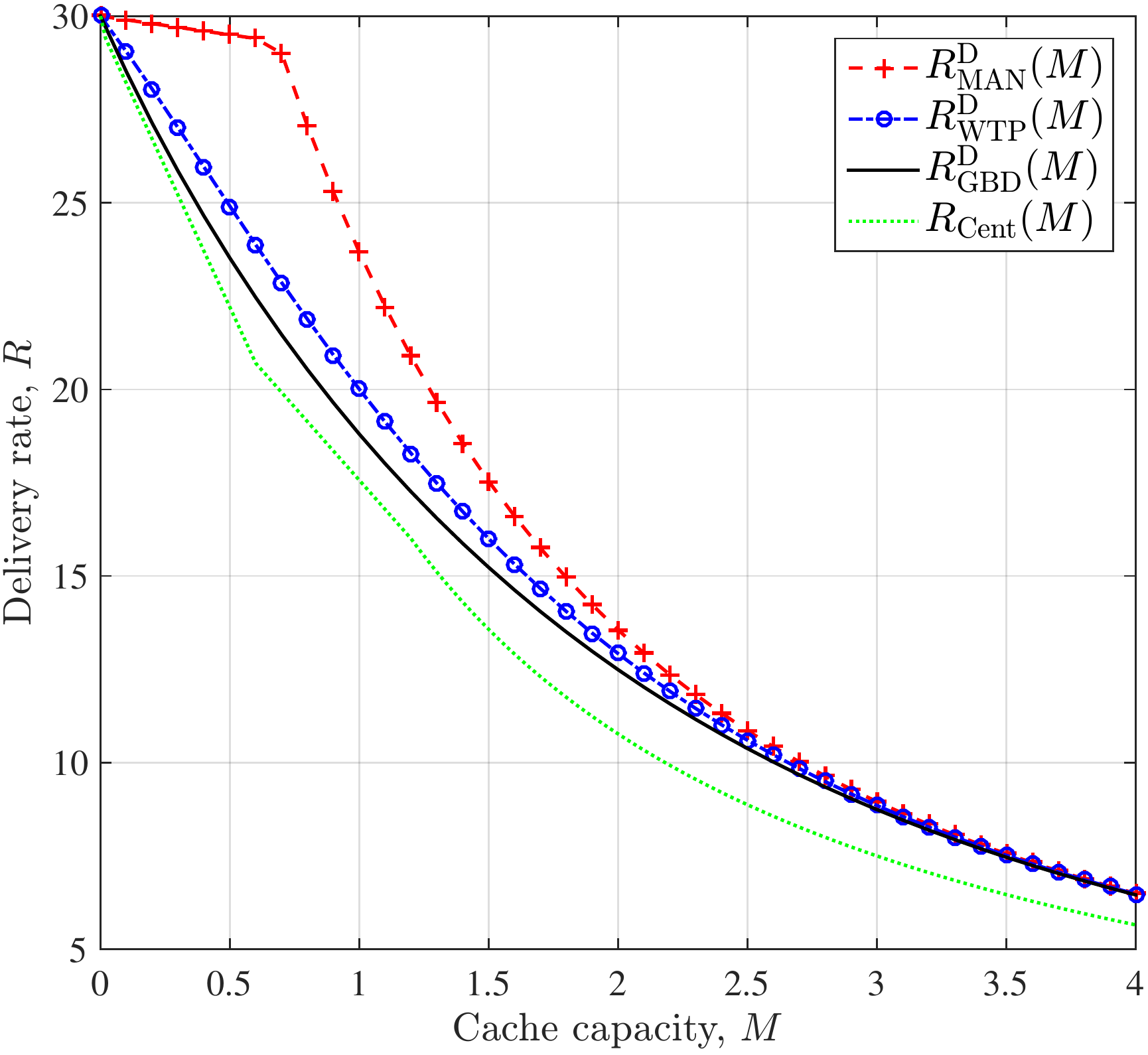}
\caption{Delivery rate-cache capacity trade-off in the decentralized setting for the GBD, MAN and WTP schemes, for $N=30$ and $K=50$. In this scenario, we have $\hat t = t^* = 2$, and the best centralized performance can be achieved by memory-sharing between the CFL scheme for $M = 1/K$, the GBC scheme for $M = N/K$, and the MAN scheme for cache capacities $M = lN/K$, where $l \in \left[ 2:K \right]$.} 
\label{DecentralizedN30K50}
\end{figure} 

Next, we compare the GBD scheme with the state-of-the-art caching schemes numerically. In Fig. \ref{DecentralizedN30K50}, the delivery rate of the proposed GBD scheme is compared with the decentralized caching schemes proposed in \cite{MaddahAliDecentralized} and \cite{KaiWanUncodedCaching}, referred to as MAN and WTP, respectively, for $N=30$ and $K=50$. The superiority of the proposed GBD scheme over the state-of-the-art caching schemes, especially for relatively small cache capacities, is visible in the figure. We also include in the figure the best achievable centralized coded caching scheme for this setting, the delivery rate-cache capacity trade-off of which is denoted by $R_{\rm{Cent}}(M)$. This curve is obtained through memory-sharing between the CFL scheme for cache capacity $M = 1/K$, the proposed GBC scheme for $M = N/K$, and the centralized MAN scheme for $M = lN/K$, where $l \in \left[ \hat t:K \right]$ (note that, in this scenario, $\hat t = t^* = 2$). Although the delivery rate of the optimal decentralized caching scheme might not be lower bounded by that of the considered centralized caching scheme (since the centralized caching scheme under consideration is not optimal), the difference between the delivery rates of the centralized and decentralized schemes roughly indicates the loss due to the decentralization of the coded caching schemes considered here. We observe that the delivery rate of the proposed GBD scheme is very close to the performance of the best known centralized scheme, particularly for small cache capacities. 

\section{Conclusion}\label{Conclusion}
We have considered a symmetric caching system with $K$ users and $N$ files, where each file has the same size of $F$ bits, and each user has the same cache capacity of $MF$ bits, that is, a cache that is sufficient to store $M$ files. The system considered here models wireless networks, in which the caches are filled over off-peak periods without any cost constraint or rate limitation (apart from the limited cache capacities), but without knowing the user demands; and all the user demands arrive (almost) simultaneously, and they are served simultaneously through an error-free shared link. We have proposed a novel group-based centralized (GBC) coded caching scheme for a cache capacity of $M = N/K$, which corresponds to the case in which the total cache capacity distributed across the network is sufficient to store all the files in the database. In the centralized scenario, each file is distributed evenly across all the users, and the users are grouped based on their requests. Our algorithm satisfies user demands in pairs, making sure that each transmitted information serves all the users in a group, that is, users requesting the same file. We have shown that the proposed caching scheme outperforms the best achievable schemes in the literature in terms of the delivery rate. Then the improvement has been extended to a larger range of cache capacities through memory-sharing with the existing achievable schemes in the literature. 

We have next employed the GBC scheme in the decentralized setting. Since the user identities are not known in the decentralized caching scenario, it is not possible to distribute the contents among users. Allowing the users to cache bits of contents randomly, we will have different parts of each content cached by a different set of users. Our decentralized caching scheme, GBD, adopts the proposed group-based centralized caching scheme in order to efficiently deliver those pats of the files that have been cached by only a single user. We have shown that the GBD scheme also achieves a smaller delivery rate than any other known scheme in the literature. We believe that the proposed group-based delivery scheme is particularly appropriate for non-uniform user demands, which is currently under investigation.

\appendices

\section{Proof of Theorem \ref{TheoremGBC}}\label{ProofFirstTheorem}
To prove Theorem \ref{TheoremGBC}, we first go through the coded delivery phase presented in Algorithm \ref{DeliveryCentralized}, and show that all user requests are satisfied at the end of the delivery phase. First part of this algorithm enables each user to obtain the subfiles of its requested file which are in the cache of all other users in the same group. We consider the first group, i.e., $i=1$ in line 2 of Algorithm \ref{DeliveryCentralized}, which includes the users that demand $W_1$. In this case, the XOR-ed contents $W_{1,k} \oplus W_{1,k+1}$, for $k \in \left[ {1:{K_1} - 1} \right]$, are delivered by the server. Having access to the subfile $W_{1,k}$ locally in its cache, each user $U_k$, for $k \in \left[ {1:{K_1}} \right]$, can decode all the remaining subfiles $W_{1,j}$, for $j \in \left[ {1:{K_1}} \right]\backslash \left\{ {k} \right\}$. Thus, a total number of $(K_{1} - 1)$ XOR-ed contents, each of size $\frac{F}{K}$ bits, are delivered by the server for the users in group $G_1$. Similarly, the second group ($i=2$ in line 2 of Algorithm \ref{DeliveryCentralized}), containing the users requesting file $W_2$, the XOR-ed contents $W_{2,k} \oplus W_{2,k+1}$, for $k \in \left[ {K_{1} + 1:{K_1 + K_{2}} - 1} \right]$, are sent by the server. With subfile $W_{2,k}$ available locally at user $U_k$, for $k \in \left[ {K_{1} + 1:{K_1 + K_{2}}} \right]$, user $U_k$ can obtain the missing subfiles $W_{2,j}$, for $j \in \left[ {K_{1} + 1:{K_1 + K_{2}}} \right] \backslash \left\{ {k} \right\}$. Hence, a total of $(K_2 - 1)F/K$ bits are served for the users in $G_2$, and so on so forth. Accordingly, for the users belonging to group $G_i$, $(K_i - 1)F/K$ bits are delivered by the server, for $i=1, ..., N$, and the total number of bits transmitted by the server in the first part of the coded delivery phase presented in Algorithm \ref{DeliveryCentralized} is given by
\begin{equation}\label{DeliveryRateFirstProcedure} \frac{F}{K}\sum\limits_{i = 1}^N {\left( {{K_i} - 1} \right)}  = \left( {K - N} \right)\frac{F}{K}.   
\end{equation}

In the second part of Algorithm \ref{DeliveryCentralized}, each user in group $G_i$, for $i \in \left[ 1:N \right]$, will decode the missing subfiles of its requested file, which are in the cache of users belonging to groups $j \in \left[ 1:N \right] \backslash \left\{ {i} \right\}$. We first start with $i=1$ and $j=2$ in lines 7 and 8, respectively. The XOR-ed contents $W_{1,k} \oplus W_{1,k+1}$, for $k \in \left[ {{K_1} + 1:{K_1} + {K_2} - 1} \right]$, i.e., the subfiles of $W_1$ cached by users in group $G_2$, are delivered in line 9. In line 10, the XOR-ed contents $W_{2,k} \oplus W_{2,k+1}$, for $k \in \left[ 1:{K_1} - 1 \right]$, i.e., the subfiles of $W_2$ cached by users in group $G_1$, are delivered by the server. Finally, by delivering $W_{1,K_1 + K_2} \oplus W_{2,K_1}$ in line 11, and having already decoded $W_{2,k}$ ($W_{1,k}$), each user $U_k$ in $G_1$ ($G_2$) can recover the missing subfiles of its requested file $W_1$ ($W_2$) which are in the cache of users in $G_2$ ($G_1$), for $k \in \left[ 1:{K_1} \right]$ (for $k \in \left[ {{K_1} + 1:{K_1} + {K_2}} \right]$). In this particular case, the number of bits delivered by the server in lines 9, 10, and 11 are $(K_2 - 1)F/K$, $(K_1 - 1)F/K$, and $F/K$, respectively, which adds up to a total number of $(K_1 + K_2 - 1)F/K$ bits. In a similar manner, the subfiles can be exchanged between users in groups $G_i$ and $G_j$, for $i \in \left[ 1:N-1 \right]$ and $j \in \left[ i + 1:N \right]$, by delivering a total of $(K_i + K_j - 1)F/K$ bits through sending the XOR-ed contents stated in lines 9, 10, and 11 of Algorithm \ref{DeliveryCentralized}. Hence, the total number of bits delivered by the server in the second part of the coded delivery phase is given by   
\begin{equation}\label{DeliveryRateSecondProcedure} \frac{F}{K}\sum\limits_{i = 1}^{N - 1} {\sum\limits_{j = i + 1}^N {\left( {{K_i} + {K_j} - 1} \right)} }  = \left( {N - 1} \right)\left( {K - \frac{N}{2}} \right)\frac{F}{K}.   
\end{equation}
By summing up \eqref{DeliveryRateFirstProcedure} and \eqref{DeliveryRateSecondProcedure}, the delivery rate of the GBC scheme is given by
\begin{equation}\label{DeliveryRateEndProcedure} {R^{\rm{C}}_{\rm{GBC}}}\left( {\frac{N}{K}} \right) = N - \frac{{N\left( {N + 1} \right)}}{{2K}}.
\end{equation}

\section{Proof of Inequality ${R_{\rm{GBC}}^{\rm{C}}}\left( N/K \right) \le f\left( N,K,t \right)$} 
\label{CentralizedImprovement}
In order to show that ${R_{\rm{GBC}}^{\rm{C}}}\left( N/K \right) \le R_b^{\rm{C}} \left( N/K \right)$, it suffices to prove that ${R_{\rm{GBC}}^{\rm{C}}}\left( N/K \right) \le f\left( N,K,t \right)$, $\forall t \in \left[ 1:K \right]$. Thus, for $N < K$ and $t \in \left[ 1:K \right]$, we need to determine 
\begin{equation}\label{AppendixCentralized1} N - \frac{{N\left( {N + 1} \right)}}{{2K}} \le {\frac{{\left( {N - 1} \right)\left( {K - t} \right)}}{{\left( {t + 1} \right)\left( {tN - 1} \right)}} + {N^2}\left( {1 - \frac{1}{K}} \right)\left( {\frac{{t - 1}}{{tN - 1}}} \right)}.
\end{equation} 
After some mathematical manipulations, the inequality \eqref{AppendixCentralized1} can be simplified to 
\begin{equation}\label{AppendixCentralized2} \left( {N - 1} \right)\left[ {{N^2}{t^2} + \left( {N + 1} \right)\left( {N - 2K} \right)t + 2{K^2} + N - 2KN} \right] \ge 0.
\end{equation}   
We define function $h:\left( \left[ {1:K} \right] \to \mathcal R \right)$ as follows:
\begin{equation}\label{AppendixCentralized3} h(t) \buildrel \Delta \over =  {{N^2}{t^2} + \left( {N + 1} \right)\left( {N - 2K} \right)t + 2{K^2} + N - 2KN},
\end{equation}  
and we need to show that $(N-1)h(t) \ge 0$. We prove this inequality for two following cases: $K \le N(N+1)/2$, and $K > N(N+1)/2$. First, let us start with $K \le N(N+1)/2$. We can rewrite $h(t)$ as follows:
\begin{equation}\label{AppendixCentralized4} h\left( t \right) = {\left( {K - tN} \right)^2} + \left( {{N^2} + N - 2K} \right)t + {K^2} + N - 2KN.
\end{equation} 
Since $t \ge 1$ and $K \le N(N+1)/2$, we have
\begin{align}\label{AppendixCentralized5} h\left( t \right) \ge& {\left( {K - tN} \right)^2} + {K^2} + {N^2} - 2KN + 2N + 2K = {\left( {K - tN} \right)^2} + \left( K - N \right) \left( K - N - 2 \right).
\end{align} 
For $t \in \left[ {1:K} \right]$ and $K \ge N + 1$, using \eqref{AppendixCentralized5}, it can be easily verified that 
\begin{equation}\label{AppendixCentralized6} \left( N - 1 \right) h\left( t \right) \ge \left( N - 1 \right) \left( \left( K - tN \right)^2 + \left( K - N \right) \left( K - N - 2 \right) \right) \ge 0.
\end{equation}
When $K > N(N+1)/2$, the proof is provided for $N \ge 3$. In this case, we have $K \ge 2N$. To prove that $h(t) \ge 0$, we first define a polynomial function of degree two $g: \left( \mathcal R^+ \to \mathcal R \right)$ as follows:
\begin{equation}\label{AppendixCentralized7} g(x) \buildrel \Delta \over =  {{N^2}{x^2} + \left( {N + 1} \right)\left( {N - 2K} \right)x + 2{K^2} + N - 2KN},
\end{equation}   
and then show that $g(x) \ge 0$, $\forall x > 0$. To have $g(x) \ge 0$, $\forall x > 0$, the following inequality should be held
\begin{equation}\label{AppendixCentralized8} {\left( {N + 1} \right)^2}{\left( {N - 2K} \right)^2} - 4{N^2}\left( {2{K^2} + N - 2KN} \right) \le 0,
\end{equation}   
which is equivalent to 
\begin{equation}\label{AppendixCentralized9} 4\left( {{{\left( {N + 1} \right)}^2} - 2{N^2}} \right)\left[ {{K^2} - KN - \frac{{{N^2}{{\left( {N - 1} \right)}^2}}}{{4\left( {2{N^2} - {{\left( {N + 1} \right)}^2}} \right)}}} \right] \le 0.
\end{equation}  
For $N \ge 3$, we have
\begin{equation}\label{AppendixCentralized10} {\left( {N + 1} \right)^2} - 2{N^2} = 1 - N\left( {N - 2} \right) < 0.
\end{equation}  
As a result, inequality \eqref{AppendixCentralized9} can be interpreted as
\begin{equation}\label{AppendixCentralized11} {{K^2} - KN - \frac{{{N^2}{{\left( {N - 1} \right)}^2}}}{{4\left( {2{N^2} - {{\left( {N + 1} \right)}^2}} \right)}}} \ge 0.
\end{equation}   
It can be easily verified that 
\newcommand\myfirstinequality{\mathrel{\overset{\makebox[0pt]{\mbox{\normalfont\tiny\sffamily (a)}}}{\le}}}
\newcommand\mysecondinequality{\mathrel{\overset{\makebox[0pt]{\mbox{\normalfont\tiny\sffamily (b)}}}{\le}}}
\begin{equation}\label{AppendixCentralized12}
\frac{{{N^2}{{\left( {N - 1} \right)}^2}}}{{4\left( {2{N^2} - {{\left( {N + 1} \right)}^2}} \right)}} \myfirstinequality 2{N^2} \mysecondinequality KN,
\end{equation}
where the inequalities (a) and (b) come from the fact that $N \ge 3$ and $K \ge 2N$, respectively. Using inequality \eqref{AppendixCentralized12} in \eqref{AppendixCentralized11}, we have
\begin{equation}\label{AppendixCentralized13} {K^2} - KN - \frac{{{N^2}{{\left( {N - 1} \right)}^2}}}{{4\left( {2{N^2} - {{\left( {N + 1} \right)}^2}} \right)}} \ge {K^2} - 2KN \ge 0,
\end{equation} 
where the first inequality is due to the fact that $K \ge 2N$. Hence, when $N \ge 3$ and $K \ge N(N+1)/2$, we have $g(x) \ge 0$, $\forall x>0$, which we conclude $h(t) \ge 0$.

\bibliographystyle{IEEEtran}
\bibliography{Report}

\begin{thebibliography}{10}
\providecommand{\url}[1]{#1}
\csname url@samestyle\endcsname
\providecommand{\newblock}{\relax}
\providecommand{\bibinfo}[2]{#2}
\providecommand{\BIBentrySTDinterwordspacing}{\spaceskip=0pt\relax}
\providecommand{\BIBentryALTinterwordstretchfactor}{4}
\providecommand{\BIBentryALTinterwordspacing}{\spaceskip=\fontdimen2\font plus
\BIBentryALTinterwordstretchfactor\fontdimen3\font minus
  \fontdimen4\font\relax}
\providecommand{\BIBforeignlanguage}[2]{{%
\expandafter\ifx\csname l@#1\endcsname\relax
\typeout{** WARNING: IEEEtran.bst: No hyphenation pattern has been}%
\typeout{** loaded for the language `#1'. Using the pattern for}%
\typeout{** the default language instead.}%
\else
\language=\csname l@#1\endcsname
\fi
#2}}
\providecommand{\BIBdecl}{\relax}
\BIBdecl

\bibitem{AlmerothCacing}
K.~C. Almeroth and M.~H. Ammar, ``The use of multicast delivery to provide a
  scalable and interactive video-on-demand service,'' \emph{{IEEE} J. Sel.
  Areas Commun.}, vol.~14, no.~6, pp. 1110--1122, Aug. 1996.

\bibitem{GolrezaeiFemtocaching}
N.~Golrezaei, A.~F. Molisch, A.~G. Dimakis, and G.~Caire, ``Femtocaching and
  device-to-device collaboration: A new architecture for wireless video
  distribution,'' \emph{{IEEE} Commun. Mag.}, vol.~51, no.~4, pp. 142--149,
  Apr. 2013.

\bibitem{MaddahAliCentralized}
M.~A. Maddah-Ali and U.~Niesen, ``Fundamental limits of caching,'' \emph{{IEEE}
  Trans. Inform. Theory}, vol.~60, no.~5, pp. 2856--2867, May 2014.

\bibitem{GregoryDtoD}
M.~Gregori, J.~Gomez-Vilardebo, J.~Matamoros, and D.~G\"und\"uz, ``Wireless
  content caching for small cell and {D}2{D} networks,'' \emph{{to appear, IEEE
  J. Sel. Areas Commun.}}, 2016.

\bibitem{JiArXivNonuniform}
M.~Ji, A.~M. Tulino, J.~Llorca, and G.~Caire, ``Order-optimal rate of caching
  and coded multicasting with random demands,'' \emph{{arXiv: 1502.03124v1
  [cs.IT]}}, Feb. 2015.

\bibitem{baev2008approximation}
I.~Baev, R.~Rajaraman, and C.~Swamy, ``Approximation algorithms for data
  placement problems,'' \emph{SIAM Journal on Computing}, vol.~38, no.~4, pp.
  1411--1429, 2008.

\bibitem{BorstCaching}
S.~Borst, V.~Gupta, and A.~Walid, ``Distributed caching algorithms for content
  distribution networks,'' in \emph{Proc. {IEEE} INFOCOM}, San Diego, CA, Mar.
  2010, pp. 1--9.

\bibitem{Blasco:ISIT:14}
P.~Blasco and D.~Gunduz, ``Multi-armed bandit optimization of cache content in
  wireless infostation networks,'' in \emph{Proc. {IEEE} Int'l Symp. on Inform.
  Theory (ISIT)}, Honolulu, HI, Jun. 2014, pp. 51--55.

\bibitem{ZhiChenXOR}
Z.~Chen, P.~Fan, and K.~B. Letaief, ``Fundamental limits of caching: Improved
  bounds for small buffer users,'' \emph{{arXiv: 1407.1935v2 [cs.IT]}}, Jul.
  2014.

\bibitem{KaiWanUncodedCaching}
K.~Wan, D.~Tuninetti, and P.~Piantanida, ``On caching with more users than
  files,'' \emph{{arXiv: 1601.063834v2 [cs.IT]}}, Jan. 2016.

\bibitem{MohammadDenizTCom}
M.~M. Amiri and D.~G\"und\"uz, ``Fundamental limits of caching: Improved
  delivery rate-cache capacity trade-off,'' \emph{{arXiv:1604.03888v1
  [cs.IT]}}, Apr. 2016.

\bibitem{SenguptaCaching}
A.~Sengupta, R.~Tandon, and T.~C. Clancy, ``Improved approximation of
  storage-rate tradeoff for caching via new outer bounds,'' in \emph{Proc.
  {IEEE} Int'l Symp. on Inform. Theory (ISIT)}, Hong Kong, Jun. 2015, pp.
  1691--1695.

\bibitem{TianCaching}
C.~Tian, ``A note on the fundamental limits of coded caching,'' \emph{{arXiv:
  1503.00010v1 [cs.IT]}}, Feb. 2015.

\bibitem{GhasemiCachingLowerBound}
H.~Ghasemi and A.~Ramamoorthy, ``Improved lower bounds for coded caching,'' in
  \emph{Proc. IEEE Int'l Symp. on Inform. Theory (ISIT)}, Hong Kong, Jun. 2015,
  pp. 1696--1700.

\bibitem{MaddahAliDecentralized}
M.~A. Maddah-Ali and U.~Niesen, ``Decentralized caching attains order optimal
  memory-rate tradeoff,'' \emph{{IEEE/ACM} Trans. Netw}, vol.~23, no.~4, pp.
  1029--1040, Apr. 2014.

\bibitem{PedarsaniOnlineCaching}
R.~Pedarsani, M.~A. Maddah-Ali, and U.~Niesen, ``Online coded caching,'' in
  \emph{Proc. {IEEE} Int'l Conf. Commun. (ICC)}, Sydney, Australia, Jun. 2014,
  pp. 1878--1883.

\bibitem{KaramchandaniHierarchical}
N.~Karamchandani, U.~Niesen, M.~A. Maddah-Ali, and S.~Diggavi, ``Hierarchical
  coded caching,'' in \emph{Proc. {IEEE} Int'l Symp. on Inform. Theory (ISIT)},
  Honolulu, HI, Jun. 2014, pp. 2142--2146.

\bibitem{NiesenNonuniform}
U.~Niesen and M.~A. Maddah-Ali, ``Coded caching with nonuniform demands,'' in
  \emph{Proc. {IEEE} Conf. Comput. Commun. Workshops (INFOCOM WKSHPS)},
  Toronto, ON, Apr. 2014, pp. 221--226.

\bibitem{ZhangDistinctFileSizes}
J.~Zhang, X.~Lin, C.~C. Wang, and X.~Wang, ``Coded caching for files with
  distinct file sizes,'' in \emph{Proc. {IEEE} Int'l Symp. on Inform. Theory
  (ISIT)}, Hong Kong, Jun. 2015, pp. 1686--1690.

\bibitem{WangHeterogenous}
S.~Wang, W.~Li, X.~Tian, and H.~Liu, ``Coded caching with heterogenous cache
  sizes,'' \emph{{arXiv:1504.01123v3 [cs.IT]}}, Apr. 2015.

\bibitem{QianDenizLossy}
Q.~Yang and D.~G\"und\"uz, ``Centralized coded caching for heterogeneous lossy
  requests,'' \emph{{[online]. Available.
  http://www.iis.ee.ic.ac.uk/dgunduz/Papers/Conference/ISIT16a.pdf}}.

\bibitem{ElzaDistortionMemoryTradeoff}
P.~Hassanzadeh, E.~Erkip, J.~Llorca, and A.~Tulino, ``Distortion-memory
  tradeoffs in cache-aided wireless video delivery,'' \emph{{arXiv:1511.03932v1
  [cs.IT]}}, Nov. 2015.

\bibitem{TimoDistortionCaching}
R.~Timo, S.~B. Bidokhti, M.~Wigger, , and B.~Geiger, ``A rate-distortion
  approach to caching,'' in \emph{Proc. International Zurich Seminar on
  Communications}, Zurich, Switzerland, Mar. 2016.

\bibitem{MaddahAliInterferenceJournal}
M.~Maddah-Ali and U.~Niesen, ``Cache-aided interference channels,''
  \emph{{arXiv: 1510.06121v1 [cs.IT]}}, Oct. 2015.

\bibitem{NaderializadehMaddahAliInterference}
N.~Naderializadeh, M.~Maddah-Ali, and A.~S. Avestimehr, ``Fundamental limits of
  cache-aided interference management,'' \emph{{arXiv:1602.04207v1 [cs.IT]}},
  Feb. 2016.

\bibitem{TimoErasureChannel}
R.~Timo and M.~Wigger, ``Joint cache-channel coding over erasure broadcast
  channels,'' \emph{{arXiv:1505.01016v1 [cs.IT]}}, May 2015.

\bibitem{HuangFadingChannelcodedcaching}
W.~Huang, S.~Wang, L.~Ding, F.~Yang, and W.~Zhang, ``The performance analysis
  of coded cache in wireless fading channel,'' \emph{{arXiv:1504.01452v1
  [cs.IT]}}, Apr. 2015.

\end{thebibliography}

\end{document}